\documentclass[12pt,letter,final]{article} %a4paper
\usepackage{comment}
\usepackage{xr}
\usepackage{amsmath}
\usepackage{amssymb}
\usepackage{amsfonts}
\usepackage{amsthm}
\usepackage{graphicx} 
\usepackage{geometry}
\usepackage{setspace}
\usepackage{hyperref}
\usepackage{url}
\usepackage{pgfplots}
\usepackage{bbm}
\usepackage{boxhandler}
\usepackage[gen]{eurosym}

\newcommand{\argmin}{\operatornamewithlimits{arg\ min}}

\usepackage[authoryear]{natbib}
\usepackage{multirow}
\bibpunct{(}{)}{;}{a}{,}{,}
\theoremstyle{plain}

\newtheorem{pro}{Proposition}
\newtheorem{cor}{Corollary}

\theoremstyle{definition}

\newtheorem{rem}{Remark}

\title{Allowing for weak identification when testing GARCH-X type models}

\author{Philipp Ketz\footnote{Paris School of Economics - CNRS. E-mail: \href{mailto:philipp.ketz@psemail.eu}{philipp.ketz@psemail.eu}.}}

\begin{document}

\maketitle

\begin{abstract}

In this paper, we use the results in \cite{AC1}, extended to allow for parameters to be near or at the boundary of the parameter space, to derive the asymptotic distributions of the two test statistics that are used in the two-step (testing) procedure proposed by \cite{PedersenRahbek:19}. The latter aims at testing the null hypothesis that a GARCH-X type model, with exogenous covariates (X), reduces to a standard GARCH type model, while allowing the ``GARCH parameter" to be unidentified. We then provide a characterization result for the asymptotic size of any test for testing this null hypothesis before numerically establishing a lower bound on the asymptotic size of the two-step procedure at the 5\% nominal level. This lower bound exceeds the nominal level, revealing that the two-step procedure does not control asymptotic size. In a simulation study, we show that this finding is relevant for finite samples, in that the two-step procedure can suffer from overrejection in finite samples. We also propose a new test that, by construction, controls asymptotic size and is found to be more powerful than the two-step procedure when the ``ARCH parameter'' is ``very small'' (in which case the two-step procedure underrejects).

\bigskip
%\noindent JEL Classification: C10, C12.
\end{abstract}

\thispagestyle{empty}

\newpage

\clearpage
\setcounter{page}{1}

\section{Introduction}

In a GARCH-X type model, the variance of a generalized autoregressive conditional heteroskedasticity (GARCH) type model is augmented by a set of exogenous regressors (X). Naturally, the question arises if the more general GARCH-X type model can be reduced to the simpler GARCH type model. Statistically speaking, the problem reduces to testing whether the coefficients on the exogenous regressors are equal to zero. As noted in \cite{PedersenRahbek:19} (PR hereinafter), the testing problem is non-standard due to the presence of two nuisance parameters that could possibly be at the boundary of the parameter space. In addition, under the null hypothesis, one of the nuisance parameters, the ``GARCH parameter'', is not identified when the other, the ``ARCH parameter'', is at the boundary. In order to address this possible lack of identification, PR suggest a two-step (testing) procedure, where rejection in the first step is taken as ``evidence'' that the model is identified. In the second step, the authors then impose an ``additional assumption'', which implies that a specific entry of the inverse information equals zero, to obtain an asymptotic null distribution of their second-step test statistic that is nuisance parameter free. There are two potential problems with this two-step procedure. First, it may not control (asymptotic) size, i.e., its (asymptotic) size may exceed the nominal level, for reasons similar to those that invalidate ``naive'' post-model-selection inference \citep[see e.g.,][]{LP:05,LP:08}. In addition, the aforementioned ``additional assumption'' may not be satisfied, which may possibly aggravate the problem. Second, the two-step procedure may, due to its two-step nature and despite the possible lack of (asymptotic) size control, have poor power in certain parts of the parameter space, as suggested by simulations in PR.\footnote{The simulation results in Appendix D of PR show that the two-step procedure has a null rejection frequency below the nominal level for ``very small'' values of the ARCH parameter; see also Section \ref{MC}.} 

In this paper, we use the results in \cite{AC1} (AC hereinafter), extended to allow for parameters to be near or at the boundary of the parameter space,\footnote{Here, the parameter space is equal to a product space; see \cite{Cox:22} for related results in the context of more general shapes of the parameter space.} to derive the asymptotic distributions of the two test statistics used in PR under weak, semi-strong, and strong identification (using the terminology in AC). These asymptotic distribution results, in turn, allow us characterize the asymptotic size of any test for testing the null hypothesis that the coefficients on the exogenous regressors are equal to zero. We numerically establish lower bounds on the asymptotic sizes of the two-step procedure proposed by PR as well as a second testing procedure proposed by PR that assumes that the ARCH parameter is known to be in the interior of the parameter space. These bounds are given by 6.65\% and 9.48\%, respectively, for a 5\% nominal level, which implies that the two testing procedures do not control asymptotic size (at the 5\% nominal level).\footnote{These lower bounds (also) apply if the ``additional assumption'' and, in case of the second testing procedure, the assumption that the ARCH parameter is in the interior of the parameter space are satisfied; see Remark \ref{} for details.} Furthermore, we propose a new test based on the second-step test statistic of PR that uses plug-in least favorable configuration critical values and, thus by construction, controls asymptotic size. 

In a small simulation study, we find that our asymptotic theory provides good approximations to the finite-sample behaviors of the tests, or testing procedures, that we consider. In particular, we find that the testing procedures proposed by PR can suffer from overrejection in finite samples. Furthermore, we find that our new test has greater power than the two-step procedure for ``very small'' values of the ARCH parameter, a presumably empirically important region of the parameter space. This finding is line with the intuition that the two-step procedure, in some sense, ``sacrifices'' power for such parameter constellations due to its two-step nature.

The remainder of this paper is organized as follows. In Section \ref{TP}, we introduce the testing problem as well as the two testing procedures proposed by PR. In Section \ref{AT}, we present the asymptotic distribution results. Section \ref{AsySz} presents the characterization result for asymptotic size and obtains the lower bounds on the asymptotic sizes of the two testing procedures proposed by PR. It also introduces our new test. The result of our simulation study are presented in Section \ref{MC}. Additional material, including proofs, is relegated to the Appendix.

Throughout this paper, we use the following conventions. All limits are taken ``as $n \to \infty$''. $e_i$ denotes a vector of zeros (of suitable dimension) with a one in the $i^\text{th}$ position. For any matrix $A$, $A_{ij}$ denotes the entry with row index $i$ and column index $j$. Furthermore, $X_n(\pi) = o_{p \pi}(1)$ means that $\sup_{\pi \in \Pi} \| X_n(\pi) \| = o_p(1)$, where $\| \cdot \|$ denotes the Euclidean norm. Lastly, ``for all $\delta_n \to 0$'' abbreviates ``for all sequences of positive scalar constants $\{ \delta_n: n \geq 1\}$ for which $\delta_n \to 0$''. 

\section{Testing problem} \label{TP}
For ease of exposition, we consider a simple version of the GARCH-X(1,1) model with a single exogenous variable (as in PR). In particular, the model is given by 
\begin{equation} \label{y}
	y_t = h_t(\theta)^{1/2} z_t,
\end{equation}
where $\theta = (\psi',\pi)' = (\beta',\zeta,\pi)'$ and
\begin{equation} \label{h}
	h_t(\theta) = h_t(\psi,\pi) = h_t(\beta,\zeta,\pi) = \zeta(1-\pi) + \beta_1 y_{t-1}^2 + \pi h_{t-1}(\psi,\pi) + \beta_2 x^2_{t-1}
\end{equation}
with $h_0(\theta) = \zeta$.\footnote{For ease of reference, we adopt the notation in AC, where $\beta$ governs the identification strength of $\pi$ (see \eqref{verification_assumption_A} below) and $\psi$ is always identified.} Here, $\{y_t,x_t\}_{t=0}^n$ is observed and $\{z_t\}_{t=0}^n$ is unobserved. The true parameter space for $\theta$, i.e., the space of all possible true values of $\theta$, is given by $\Theta^* = \Psi^* \times \Pi^*$, where
\[
	\Psi^* = \{ \psi : 0 \leq \beta_1 \leq \overline{\beta}^*_1, 0 \leq \beta_2 \leq \overline{\beta}^*_2, \underline{\zeta}^* \leq \zeta \leq \overline{\zeta}^*  \} \text{ and } \Pi^* = \{ \pi : 0 \leq \pi \leq \overline{\pi}^* \}
\]
for some $0 < \overline{\beta}^*_1 < \infty$, $0 < \overline{\beta}^*_2 < \infty$, $0 < \underline{\zeta}^* < \overline{\zeta}^* < \infty $, and $0< \overline{\pi}^* < 1$. 

The model is estimated by quasi-maximum likelihood. In particular, the objective function is given by (-$\frac{1}{n}$ times) the Gaussian-based conditional quasi log-likelihood function, i.e., $Q_n(\theta) = \frac{1}{n} \sum_{t=1}^n l_t(\theta)$, where 
\[
	l_t(\theta) = \frac{1}{2} \log(2 \tilde{\pi} )  + \frac{1}{2} \log(h_t(\theta)) + \frac{y_t^2}{2h_t(\theta)}
\]
and where $\tilde{\pi} = 3.14...$. The quasi-maximum likelihood estimator is given by
\[
	\hat{\theta}_n= \argmin_{\theta \in \Theta} Q_n(\theta),
\]
where $\Theta = \Psi \times \Pi$ denotes the optimization parameter space with
\[
	\Psi = \{ \psi : 0 \leq \beta_1 \leq \overline{\beta}_1, 0 \leq \beta_2 \leq \overline{\beta}_2, \underline{\zeta} \leq \zeta \leq \overline{\zeta}  \} \text{ and } \Pi = \{ \pi : 0 \leq \pi \leq \overline{\pi} \}
\]
for some $\overline{\beta}^*_1  < \overline{\beta}_1< \infty$, $\overline{\beta}^*_2  < \overline{\beta}_2< \infty$, $0 < \underline{\zeta} < \underline{\zeta}^*$, $\overline{\zeta}^* < \overline{\zeta} < \infty$, and $\overline{\pi}^*< \overline{\pi} < 1$. Note that, given the definitions of $\Theta^*$ and $\Theta$, (the true values of) $\beta_1$, $\beta_2$, and $\pi$ are allowed to be at the boundary of the optimization parameter space.

While our asymptotic distribution results are useful for analyzing a wide range of testing problems, we are mainly interested in testing
\begin{equation} \label{testing_problem}
	H_0: \beta_2 = 0 \text{\ vs.\ } H_1:\beta_2 > 0.
\end{equation}
As pointed out in PR, this testing problem is non-standard in that, under $H_0$, there are two nuisance parameters that \textit{may} be at the boundary of the (optimization) parameter space, $\beta_1$ and $\pi$. Furthermore, when $\beta_1$ is at the boundary ($\beta_1 = 0$) then $\pi$ is not identified under $H_0$. To see this, note that, given $h_0(\theta) = \zeta$, we have 
\begin{equation} \label{verification_assumption_A}
	h_t(\theta) = \zeta +  \beta_1 \sum_{i=0}^{t-1} \pi^i y_{t-i-1}^2 + \beta_2 \sum_{i=0}^{t-1} \pi^i x_{t-i-1}^2
\end{equation}
and, thus, $h_t(0,\zeta,\pi) = \zeta$ $\forall \theta = (\beta,\zeta,\pi) = (0,\zeta,\pi) \in \Theta, \forall n \geq 1$. In words, under $H_0$, $\beta_1 = 0$ implies that the distribution of the data does not depend on $\pi$, i.e., $\pi_1$ and $\pi_2$ (with $\pi_1 \neq \pi_2$) are observationally equivalent. 

Let $\mathcal{T}_n$ denote a generic test statistic for testing \eqref{testing_problem} and let cv$_{n,1-\alpha}$ denote the corresponding nominal level $\alpha$ critical value, which may depend on $n$. The size of the test that rejects $H_0$ when $\mathcal{T}_n > \text{cv}_{n,1-\alpha}$ is given by Sz$_\mathcal{T} = \sup_{\gamma \in \Gamma:\beta_2 = 0} P_\gamma (\mathcal{T}_n > \text{cv}_{n,1-\alpha})$, where $\Gamma$ denotes the true parameter space for $\gamma = (\theta,\phi)$ and where $\phi$ denotes the distribution of $\{ x_t, z_t \}$. We say that a test controls size if Sz$_\mathcal{T} \leq \alpha$. We note that ``uniformity'' (over $\Gamma$) is built into the definition of Sz$_\mathcal{T}$ and that whether a test controls size crucially depends on $\Gamma$. Typically, it is infeasible to compute Sz$_\mathcal{T}$. Therefore, we rely on asymptotic approximations. In particular, we approximate the (``finite-sample'') size of a test by its asymptotic size, which is given by
\[
	\text{AsySz}_\mathcal{T} = \limsup_{n\to \infty} \sup_{\gamma \in \Gamma:\beta_2 = 0} P_\gamma (\mathcal{T}_v > \text{cv}_{n,1-\alpha}).
\]
While $\text{AsySz}_\mathcal{T}$ ``still'' depends on $\Gamma$, it generally only does so through a finite-dimensional parameter, making its evaluation ``easier''. We say that a test controls asymptotic size if AsySz$_\mathcal{T} \leq \alpha$. In large samples, AsySz$_\mathcal{T}$ provides a good approximation of Sz$_\mathcal{T}$ so that a test that controls asymptotic size can be expected to ``approximately'' control size. Therefore, in what follows we focus on whether or not a given test, or testing procedure, controls asymptotic size. 

\subsection{Testing procedures proposed by PR}
Given the non-standard nature of the testing problem in \eqref{testing_problem}, PR propose a two-step procedure to deal with the presence of nuisance parameters on the boundary of the parameter space as well as the lack of identification of $\pi$ when $\beta_1 = 0$. In the first step, PR propose to test $H_0^\dagger: \beta_1 = \beta_2 = 0$ using the corresponding (rescaled) quasi-likelihood ratio statistic
\[
	LR^\dagger_n = 2n (Q_n(\hat{\theta}^\dagger_{n,0})  - Q_n(\hat{\theta}_n) ) /\hat{c}_n^\dagger,
\]
where $\hat{\theta}^\dagger_{n,0} = \argmin_{\theta \in \Theta^\dagger_0}Q_n({\theta})$ with $\Theta^\dagger_0 = \{ \theta \in \Theta: \beta_1 = \beta_2 = 0\}$ and where\footnote{\label{estimation_c}Note that other estimators of $c_0 = c(\gamma_0) = E_{\gamma_0} (z_t^2 - 1)^2/2$ are available (see Section \ref{AT} for the ``definition'' of $\gamma_0$): For example, $\hat{c}_{\text{alt}}(\theta) =  \left( \frac{1}{n} \sum_{t=1}^n  \frac{y_t^4}{h_t(\theta)^2} - 1 \right)/2$ evaluated at $\hat{\theta}^\dagger_{n,0}$ or $\hat{\theta}_{n}$; similarly, one could evaluate $\hat{c}(\theta)$ at $\hat{\theta}_{n}$. Note that \cite{Andrews:01} also considers (rescaled) quasi-likelihood ratio statistics of the form
\[
	2n  (\min_{\pi \in \Pi} Q_n(\hat{\psi}^\dagger_{n,0}(\pi),\pi)/c(\hat{\psi}_n(\pi),\pi)  - \min_{\pi \in \Pi}Q_n(\hat{\psi}_n(\pi),\pi)/c(\hat{\psi}_n(\pi),\pi) ),
\]
where $\hat{\theta}^\dagger_{n,0} = ((\hat{\psi}^\dagger_{n,0}(\hat{\pi}^\dagger_{n,0}))',\hat{\pi}^\dagger_{n,0})'$. We note that $Q_n(\hat{\psi}^\dagger_{n,0}(\pi),\pi)$ does not depend on $\pi$.
}
\[
	\hat{c}^\dagger_n = \hat{c}(\hat{\theta}^\dagger_{n,0}) =  \left( \frac{1}{n} \sum_{t=1}^n \left(  \frac{y_t^2}{h_t(\hat{\theta}^\dagger_{n,0})} - 1 \right)^2 \right)/2 = \left( \frac{1}{n} \sum_{t=1}^n \left( \frac{y_t^2}{\hat{\zeta}^\dagger_{n,0}} - 1\right)^2 \right)/2 .
\]
Here, $\hat{\theta}^\dagger_{n,0} = (\hat{\zeta}^\dagger_{n,0},0,0,\hat{\pi}^\dagger_{n,0})'$ with $\hat{\pi}^\dagger_{n,0} = \Pi$. The asymptotic null distribution of $LR^\dagger_n$ can be derived using the results in \cite{Andrews:01} (see e.g., Theorem 2.1 in PR) and, although not available in closed form, can easily be simulated from. If $H_0^*$ is not rejected, then $H_0$ is not rejected. If $H_0^*$ is rejected, then PR conclude that $\beta_1 > 0$. Given that $\beta_1 > 0$, the nuisance parameter $\pi$ is identified and the only remaining ``problem'' is that $\pi$ may still be at the boundary ($\pi = 0$). In the second step, PR then suggest to test \eqref{testing_problem}, under the maintained assumption that $\beta_1 > 0$, using the corresponding (rescaled) quasi-likelihood ratio statistic
\[
	LR_n = 2n  (Q_n({\hat{\theta}_{n,0}})  - Q_n(\hat{\theta}_n) )/\hat{c},
\]
where $\hat{\theta}_{n,0} = \argmin_{\theta \in \Theta_0}Q_n({\theta})$ with $\Theta_0 = \{ \theta \in \Theta: \beta_2 = 0\}$ and where $\hat{c}_n = \hat{c}(\hat{\theta}_{n,0})$.\footnote{Alternatively, we could evaluate $\hat{c}(\theta)$ at $\hat{\theta}_{n}$ or use $\hat{c}_{\text{alt}}(\theta)$ (see footnote \ref{estimation_c}) evaluated at $\hat{\theta}_{n,0}$ or $\hat{\theta}_{n}$.}
In this context, PR make an ``additional assumption'' on the dependence between $\{ x_t \}$ and $\{ y_t \}$ that implies that a certain entry of the inverse of the information matrix equals zero when $\pi = 0$. This, in turn, implies that the asymptotic null distribution of $LR_n$ simplifies to $\max(0,Z)^2$, where $Z\sim N(0,1)$. Formally, the two-step procedure (TS) is defined as follows: Reject $H_0$ if 
\[
	{TS}_n = \mathbbm{1}(LR_n^\dagger > \widetilde{LR}^\dagger_{n,1-\alpha}) \times LR_n > \text{cv}_{1-\alpha},
\]
where $\text{cv}_{1-\alpha}$ and $\widetilde{LR}^\dagger_{1-\alpha}$ denote the $1-\alpha$ quantiles of $\max(0,Z)^2$ and the asymptotic distribution given in \eqref{asy_dist_LR_star} with $b = 0$ and unknown quantities replaced by consistent estimators (see e.g., Appendix \ref{CD}), respectively.\footnote{PR do not formally define their two-step procedure in that they do not define the nominal levels that ought to be used in the two steps, say $\alpha_1$ and $\alpha_2$. Here, we take $\alpha_1 = \alpha_2 = \alpha$. Of course, given the results in this paper, it is in principle possible to choose $\alpha_1$ and $\alpha_2$ in order to ensure that AsySz$_{TS}\leq \alpha$. We refrain from doing so, however, because, given the results in this paper, the motivation for using a two-step procedure is rendered obsolete; see e.g., the new test.}  While it is possible to derive the asymptotic null distribution of $LR_n$ without this ``additional assumption'', PR refrain from doing so.\footnote{If $\beta_1 > 0$, or rather $\beta_1 \geq c > 0$ for some $c \in \mathbb{R}$, then it is straightforward to test \eqref{testing_problem} using the approach in \cite{Ketz:JMP}, without any assumptions on the dependence between $\{ x_t \}$ and $\{ y_t \}$.} 

PR also mention that if it is \textit{a priori} known that $\beta_1 > 0$ then one may directly test \eqref{testing_problem} using $LR_n$ and their suggested critical value, $\text{cv}_{1-\alpha}$. We refer to this test as the ``second testing procedure (of PR)''. Formally, the second testing procedure (S) is defined as follows: If $\beta_1 > 0$, reject $H_0$ if 
\[
	S_n = LR_n > \text{cv}_{1-\alpha}.
\]

There are several potential issues with the above testing procedures. First, the intuition underlying the two-step procedure is based on the assumption that the first-step test never makes a type I error. This assumption, however, cannot hold and a type I error in the first step may very well propagate to the type I error of the overall (two-step) procedure.\footnote{The only way for the first-step test to never make a type I error is to take the nominal level of the corresponding test equal to zero. This, however, would lead the first-step test (as well as the two-step procedure) to have zero power.}  Relatedly, $\beta_1$ may be close to zero relative to the sample size; see Section \ref{AT} for details. In that case, the probability of rejecting $H_0^\dagger$ exceeds the first-step nominal level and $\pi$ is only weakly identified such that the asymptotic distribution result for $LR_n$ that takes $\pi$ to be (strongly) identified may only provide a very poor approximation to its actual finite-sample distribution. The latter may also be an issue for the second testing procedure. In both cases, one may be worried that the testing procedure does not control asymptotic size. As mentioned above, asymptotic size depends on $\Gamma$. For sake of brevity, the definition of $\Gamma$ is given in Appendix \ref{ACverification}.

\section{Asymptotic distribution results} \label{AT}

As shown in the recent literature \citep[see e.g.,][AC]{AG1}, asymptotic size is intrinsically linked to the asymptotic distribution of the test statistic under (drifting) sequences of true parameters. Let $\gamma_n = (\theta_n,\phi_n) = (\beta_n,\zeta_n,\pi_n,\phi_n) \in \Gamma$ denote the true parameter for $n \geq 1$. In the context at hand, the following (sets of) sequences are key:
\begin{align*}
	\Gamma(\gamma_0) =& \{ \{ \gamma_n \in \Gamma : n \geq 1 \} : \gamma_n \to \gamma_0 = (\theta_0,\phi_0) = (\beta_0,\zeta_0,\pi_0,\phi_0) \in \Gamma \},\\
	\Gamma(\gamma_0,0,b)  = &\{ \{ \gamma_n \} \in \Gamma(\gamma_0) : \beta_0 = 0 \text{ and } \sqrt{n} \beta_n \to b \in \mathbb{R}^2_{+,\infty} \}, \text{ and } \nonumber \\ \nonumber
	\Gamma(\gamma_0,\infty,b,\omega_0,p)  =& \{ \{ \gamma_n \} \in \Gamma(\gamma_0) : \sqrt{n} \beta_n \to b \in \mathbb{R}^2_{+,\infty} \text{ with } \| b \| = \infty, \\
	\nonumber & \ \beta_n/\| \beta_n \| \to \omega_0 \in \mathbb{R}_+^2, \text{ and } \sqrt{n} \| \beta_n \| \pi_n \to p \in [0,\infty]\}.
\end{align*}
In what follows, we use the terminology ``under $\{\gamma_n\} \in \Gamma(\gamma_0)$'' to mean ``when the true parameters are $\{ \gamma_n \} \in \Gamma(\gamma_0)$ for any $\gamma_0 \in \Gamma$, ``under $\{\gamma_n\} \in \Gamma(\gamma_0,0,b)$'' to mean ``when the true parameters are $\{ \gamma_n \} \in \Gamma(\gamma_0,0,b)$ for any $\gamma_0 \in \Gamma$ with $\beta_0 = 0$ and any $b \in \mathbb{R}^2_{+,\infty}$'', and ``under $\{\gamma_n\} \in \Gamma(\gamma_0,\infty,b,\omega_0,p)$'' to mean ``when the true parameters are $\{ \gamma_n \} \in \Gamma(\gamma_0,\infty,b,\omega_0,p)$ for any $\gamma_0 \in \Gamma$, any $b \in \mathbb{R}^2_{+,\infty} \text{ with } \| b \| = \infty$, any $\omega_0 \in \mathbb{R}_+^2$ with $\| \omega_0\| = 1$ if $b_1=b_2=\infty$ or $\omega_0 = e_j  \text{ if } b_j = \infty \text{ and } b_{j'} < \infty \text{ for } j,j' \in \{ 1,2\}$ and $j'\neq j$, and any $p \in [0,\infty]$ if $\pi_0 = 0$ or $p = \infty$ if $\pi_0 >0$''. We note that under sequences of true parameters for which $\sqrt{n} \beta_n \to b$ with $\| b \| < \infty$, identification of $\pi$ is weak, while under sequences of true parameters for which $\beta_n \to 0$ but $\sqrt{n} \beta_n \to \infty$, identification of $\pi$ is semi-strong.

All claims in this section, including the following, are verified in Appendix \ref{ACverification}. Under $\{\gamma_n\} \in \Gamma(\gamma_0)$, we have $\sup_{\pi \in \Pi} \| \hat{\psi}_n(\pi) - \psi_n \| \to_p 0$ and $\| \hat{\psi}_n - \psi_n \| \to_p 0$ when $\beta_0 = 0$, while $\hat{\theta}_n - \theta_n \to_p 0$ when $\beta_0 \neq 0$.

\subsection{Results for $\beta_n \to 0$ including asymptotic distribution for $\| b \| < \infty$} \label{Results_close_to_zero}
First, we derive results for $\beta$ close to zero, i.e., $\beta_n \to 0$. To that end, let $\hat{\theta}_n = (\hat{\psi}_n(\hat{\pi}_n), \hat{\pi}_n)$, where $\hat{\psi}_n(\pi) = (\hat{\beta}_n(\pi),\hat{\zeta}_n(\pi)) \in \Psi$ is the concentrated extremum estimator of $\psi$ for given $\pi \in \Pi$, i.e.,
\[
	\hat{\psi}_n(\pi) = \argmin_{\psi \in \Psi} Q_n(\psi,\pi),
\]
and where $\hat{\pi}_n$ is the minimizer of the concentrated objective function $Q_n^c(\pi) = Q_n(\hat{\psi}_n(\pi),\pi)$, i.e.,
\[
	\hat{\pi}_n= \argmin_{\pi \in \Pi} Q_n^c(\pi).
\]
With a slight abuse of notation, we let $\hat{\pi}_n = 1$ whenever $\hat{\pi}_n = \Pi$, which occurs when $\hat{\beta}_n(\pi) = 0$ for all $\pi \in \Pi$. Under $\{\gamma_n\} \in \Gamma(\gamma_0,0,b)$, $Q_n(\psi,\pi)$ admits the following quadratic expansion in $\psi$ around $\psi_{0,n} = (0,\zeta_n)$ for given $\pi$:
\begin{align}
	Q_n(\psi,\pi) &= Q_n(\psi_{0,n},\pi) + D_\psi Q_n(\psi_{0,n},\pi)'(\psi-\psi_{0,n}) \nonumber \\ 
	&+ \frac{1}{2} (\psi-\psi_{0,n})' H(\pi;\gamma_0)(\psi-\psi_{0,n}) + R_n(\psi,\pi), \label{quadratic_expansion}
\end{align}
where the remainder $R_n(\psi,\pi)$ satisfies
\[
	\sup_{\psi \in \Psi: \| \psi - \psi_{0,n} \| \leq \delta_n} | a_n^2(\gamma_n) R_n(\psi,\pi) | = o_{p\pi}(1)
\]
for all $\delta_n \to 0$ and where
\[
	a_n(\gamma_n) = \left\{ \begin{array}{ll} n^{1/2} & \text{if } \{ \gamma_n \} \in \Gamma(\gamma_0,0,b) \text{ and } \|b\| < \infty \\ \| \beta_n \|^{-1} & \text{if } \{ \gamma_n \} \in \Gamma(\gamma_0,0,b) \text{ and } \|b\| = \infty \end{array} \right. .
\]
Here, $D_\psi Q_n(\theta)$ denotes the left/right partial derivatives of $Q^\infty_n(\theta) = \frac{1}{n} \sum_{t=1}^n l^\infty_t(\theta)$ with respect to $\psi$, where
\[
	l^\infty_t(\theta) = \frac{1}{2} \log(2 \tilde{\pi} )  + \frac{1}{2} \log(h_t^\infty(\theta)) + \frac{y_t^2}{2h^\infty_t(\theta)}
\]
and where
\[
	h_t^\infty(\theta) = \zeta +  \beta_1 \sum_{i=0}^{\infty} \pi^i y_{t-i-1}^2 + \beta_2 \sum_{i=0}^{\infty} \pi^i x_{t-i-1}^2.
\]
In particular,
\[
	 D_\psi Q_n(\theta) = \frac{1}{n} \sum_{t=1}^n l^\infty_{\psi,t}(\theta),
\]
where
\[
	l^\infty_{\psi,t}(\theta)  = \frac{\partial}{\partial \psi} l^\infty_t(\theta) = \frac{1}{2h^\infty_t(\theta)} \left( 1-  \frac{y_t^2}{h^\infty_t(\theta)} \right)  \frac{\partial h^\infty_t(\theta)}{\partial \psi}
\]
and where
\[
	\frac{\partial h^\infty_t(\theta)}{\partial \psi} = \left( \sum_{i=0}^{\infty} \pi^i y_{t-i-1}^2, \sum_{i=0}^{\infty} \pi^i x_{t-i-1}^2, 1 \right)'.
\]
$H(\pi;\gamma_0)$ is defined below. Next, define the empirical process $\{G_n(\pi): \pi \in \Pi \}$ by 
\[
	G_n(\pi) = \sqrt{n} \frac{1}{n} \sum_{t=1}^n ( l^\infty_{\psi,t}(\psi_{0,n},\pi) - E_{\gamma_n} l^\infty_{\psi,t}(\psi_{0,n},\pi)).
\]
Under $\{\gamma_n\} \in \Gamma(\gamma_0,0,b)$, $G_n(\cdot) \Rightarrow G(\cdot;\gamma_0)$, where $\Rightarrow$ denotes weak convergence and where $G(\cdot;\gamma_0)$ is a mean zero Gaussian process with bounded continuous sample paths and covariance Kernel given by
{\footnotesize
\begin{align}
	 &\Omega(\pi_1,\pi_2; \gamma_0) = \nonumber \\ &\frac{c_0}{2} \left[ \begin{array}{ccc}   \frac{2c_0}{1-\pi_1\pi_2} + \frac{1}{(1-\pi_1)(1-\pi_2)}  & \frac{1}{\zeta_0}E_{\gamma_0} \sum_{i=0}^\infty \pi_1^i z^2_{t-i-1}  \sum_{j=0}^\infty \pi_2^j x^2_{t-j-1} &  \frac{1}{\zeta_0} \frac{1}{1-\pi_1} \\  \frac{1}{\zeta_0}E_{\gamma_0} \sum_{i=0}^\infty \pi_1^i x^2_{t-i-1}  \sum_{j=0}^\infty \pi_2^j z^2_{t-j-1} &  \frac{1}{\zeta_0^2} E_{\gamma_0} \sum_{i=0}^\infty \pi_1^i x^2_{t-i-1} \sum_{j=0}^\infty \pi_2^j x^2_{t-j-1} &  \frac{1}{\zeta_0^2}E_{\gamma_0}   \sum_{j=0}^\infty \pi_1^j x^2_{t-j-1} \\  \frac{1}{\zeta_0} \frac{1}{1-\pi_2} & \frac{1}{\zeta_0^2}E_{\gamma_0}   \sum_{j=0}^\infty \pi_2^j x^2_{t-j-1} & \frac{1}{\zeta_0^2} \end{array} \right] \label{omega}
\end{align}}for $\pi_1, \pi_2 \in \Pi$, where $c_0 = c(\gamma_0) = \frac{E_{\gamma_0}(z_t^2-1)^2}{2} =  \frac{E_{\gamma_0}z_t^4-1}{2} $.\footnote{The last equality follows by the definition of $\Gamma$.} Furthermore, we have that 
\[
	H(\pi;\gamma_0) = \Omega(\pi,\pi; \gamma_0)/c_0
\]
for $\pi \in \Pi$. Let
\[
	Z_n(\pi;\gamma_0) = -a_n(\gamma_n) H^{-1}(\pi;\gamma_0)D_\psi Q_n(\psi_{0,n},\pi).
\]
Then, under $\{\gamma_n\} \in \Gamma(\gamma_0,0,b)$, we have
\begin{equation} \label{dist_Z}
	Z_n(\pi;\gamma_0) \overset{d}{\to} \begin{cases} Z(\pi;\gamma_0,b) & \text{if } \|b\| < \infty \\
	- H^{-1}(\pi;\gamma_0)  K(\pi;\gamma_0) \omega_0 &  \text{if } \|b\| = \infty  \text{ and } \beta_n/\| \beta_n \| \to \omega_0
\end{cases},
\end{equation}
where 
\[
	Z(\pi;\gamma_0,b) = - H^{-1}(\pi;\gamma_0) \{ G(\pi;\gamma_0) + K(\pi;\gamma_0)b \}
\]
and where
\[
	 K(\pi;\gamma_0) = - \frac{1}{2} \left[ \begin{array}{cc}   \frac{2c}{1-\pi\pi_0} + \frac{1}{(1-\pi)(1-\pi_0)}  &  \frac{1}{\zeta_0}E_{\gamma_0} \sum_{i=0}^\infty \pi^i z^2_{t-i-1}  \sum_{j=0}^\infty \pi_0^j x^2_{t-j-1} \\  \frac{1}{\zeta_0}E_{\gamma_0} \sum_{i=0}^\infty \pi^i x^2_{t-i-1}  \sum_{j=0}^\infty \pi_0^j z^2_{t-j-1} & \frac{1}{\zeta_0^2} E_{\gamma_0} \sum_{i=0}^\infty \pi^i x^2_{t-i-1} \sum_{j=0}^\infty \pi^j_0 x^2_{t-j-1} \\  \frac{1}{\zeta_0} \frac{1}{1-\pi_0} &\frac{1}{\zeta_0^2}E_{\gamma_0}   \sum_{j=0}^\infty \pi_0^j x^2_{t-j-1} \end{array} \right].
\]
Next, let
\[
	q_n(\lambda,\pi;\gamma_0) = (\lambda - Z_n(\pi;\gamma_0))'H(\pi;\gamma_0)(\lambda - Z_n(\pi;\gamma_0))
\]
such that the quadratic expansion in \eqref{quadratic_expansion} can be written as
\begin{align*}
	Q_n(\psi,\pi) &= Q_n(\psi_{0,n},\pi) - \frac{1}{2a^2_n(\gamma_n)}Z_n(\pi;\gamma_0)'H(\pi;\gamma_0)Z_n(\pi;\gamma_0)\\
	&+  \frac{1}{2a^2_n(\gamma_n)} q_n(a_n(\gamma_n)(\psi-\psi_{0,n}),\pi;\gamma_0) + R^*_n(\psi,\pi).
\end{align*}
Note that the first two terms do not depend on $\psi$. As a result, following \cite{Andrews:99,Andrews:01}, it can then be shown that the (scaled and demeaned) minimizer of $Q_n(\psi,\pi)$ with respect to $\psi$ for given $\pi$ asymptotically behaves like the minimizer of the ``asymptotic version'' of $q_n(\cdot)$ over an appropriately defined parameter space for $\lambda$. 

\subsubsection{Asymptotic distribution for $\| b \| < \infty$} \label{ad_b_less_infty}

In particular, under $\{\gamma_n\} \in \Gamma(\gamma_0,0,b)$ with $\| b \| < \infty$, we have 
\begin{equation} \label{lambda_hat}
	\sqrt{n} (\hat{\psi}_n(\pi) -\psi_{0,n}) \overset{d}{\to} \hat{\lambda}(\pi;\gamma_0,b),
\end{equation}
where 
\[
	\hat{\lambda}(\pi;\gamma_0,b) = \argmin_{\lambda \in \Lambda} q(\lambda,\pi;\gamma_0,b).
\]
Here,
\[
 	q(\lambda,\pi;\gamma_0,b) =  (\lambda - Z(\pi;\gamma_0,b))'H(\pi;\gamma_0)(\lambda - Z(\pi;\gamma_0,b))
\]
and
\[
	 \Lambda = [0,\infty]^2 \times [-\infty,\infty].
\]
Next, let 
\[
	\hat{\pi}(\gamma_0,b) = \argmin_{\pi \in \Pi} q(\hat{\lambda}(\pi;\gamma_0,b),\pi;\gamma_0,b)
\]
and, with a slight abuse of notation, let $\hat{\pi}(\gamma_0,b) = 1$ whenever $\hat{\pi}(\gamma_0,b) = \Pi$. Then, we have the following first main result of this paper. Namely, under $\{\gamma_n\} \in \Gamma(\gamma_0,0,b)$ with $\| b \| < \infty$, we have
\begin{equation} \label{asy_dist_estimator}
	\left( \begin{array}{c} \sqrt{n} (\hat{\psi}_n -\psi_{0,n}) \\ \hat{\pi}_n \end{array} \right) \overset{d}{\to} \left( \begin{array}{c} \hat{\lambda}(\hat{\pi}(\gamma_0,b);\gamma_0,b) \\ \hat{\pi}(\gamma_0,b) \end{array} \right)
\end{equation}
and\footnote{Note that $\min_{\pi \in \Pi} - \hat{\lambda}(\pi;\gamma_0,b)'H(\pi;\gamma_0)\hat{\lambda}(\pi;\gamma_0,b) = - \hat{\lambda}(\hat{\pi}(\gamma_0,b);\gamma_0,b)'H(\pi;\gamma_0)\hat{\lambda}(\hat{\pi}(\gamma_0,b);\gamma_0,b)$.}
\begin{equation} \label{asy_dist_objective}
	2n(Q_n(\hat{\theta}_n) - Q_{0,n}) \overset{d}{\to} \min_{\pi \in \Pi} - \hat{\lambda}(\pi;\gamma_0,b)'H(\pi;\gamma_0)\hat{\lambda}(\pi;\gamma_0,b).
\end{equation}
The corresponding asymptotic distribution results for $\hat{\theta}^\dagger_{n,0}$ and $\hat{\theta}_{n,0}$ are obtained similarly, by replacing $\Lambda$ with $\{0\}^2 \times [-\infty,\infty]$ and $ [0,\infty] \times \{0\} \times [-\infty,\infty]$, respectively.

\begin{figure}[h!] 
  \begin{center}
    \includegraphics[width=54mm]{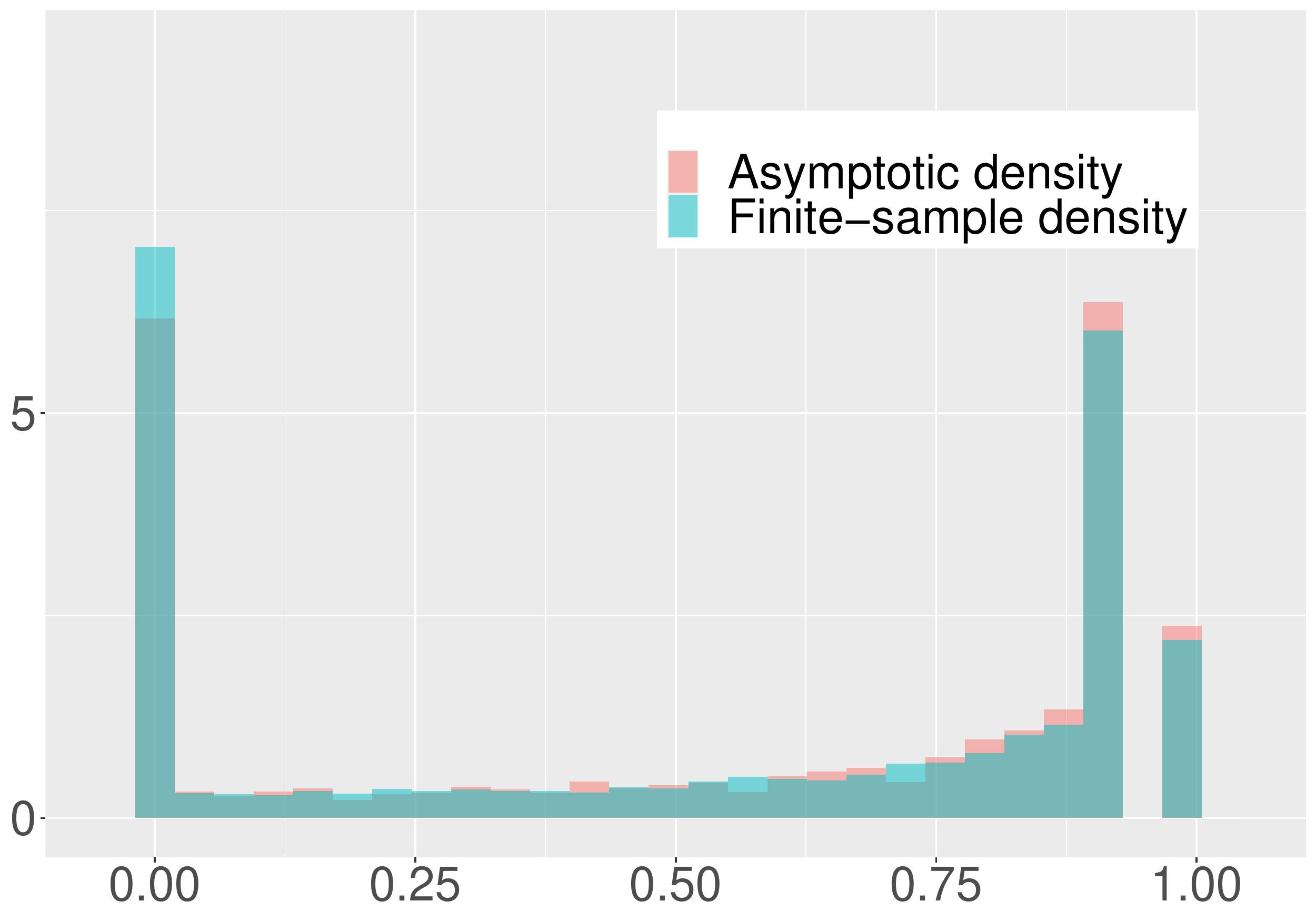}
    \includegraphics[width=54mm]{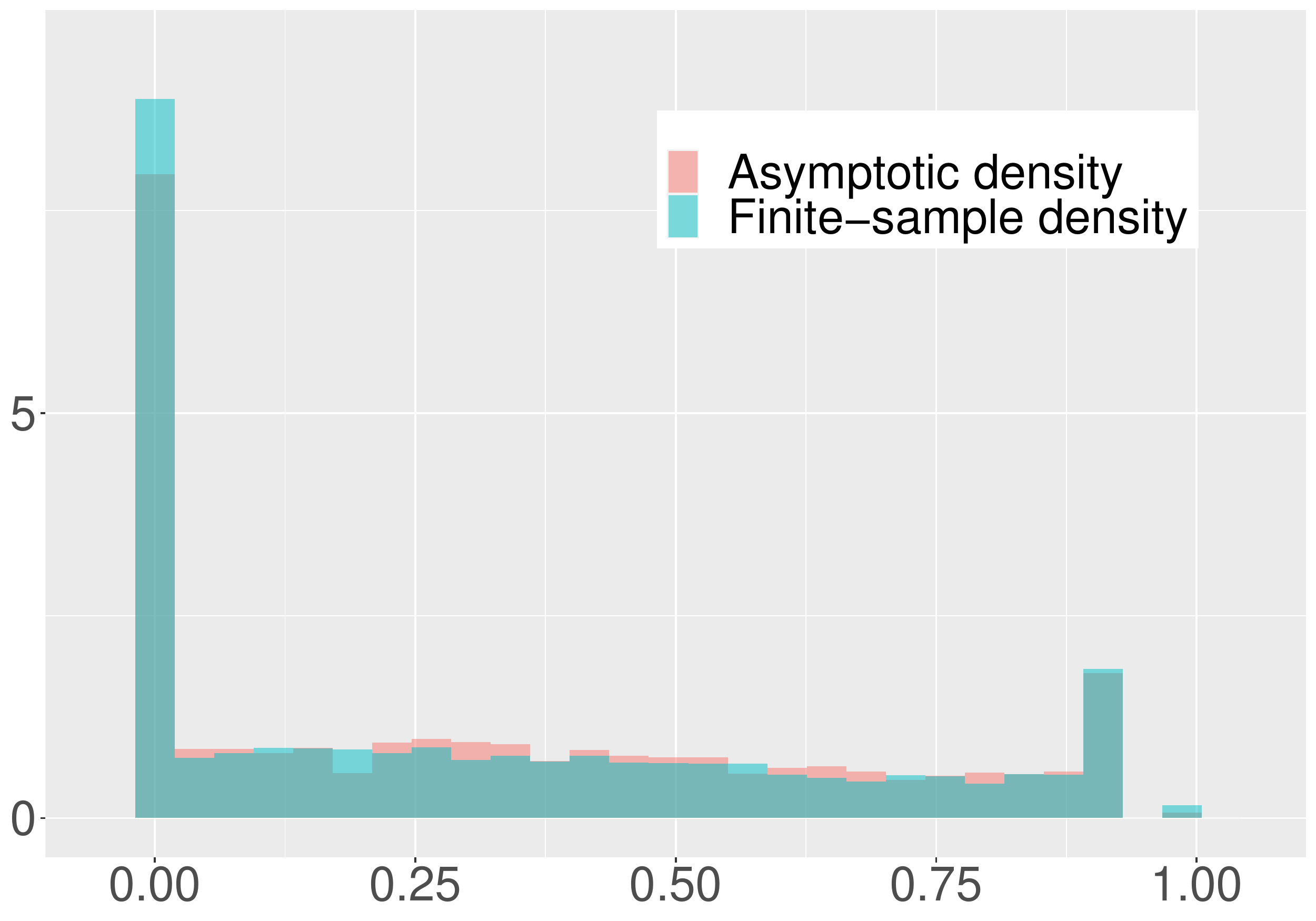}
    \includegraphics[width=54mm]{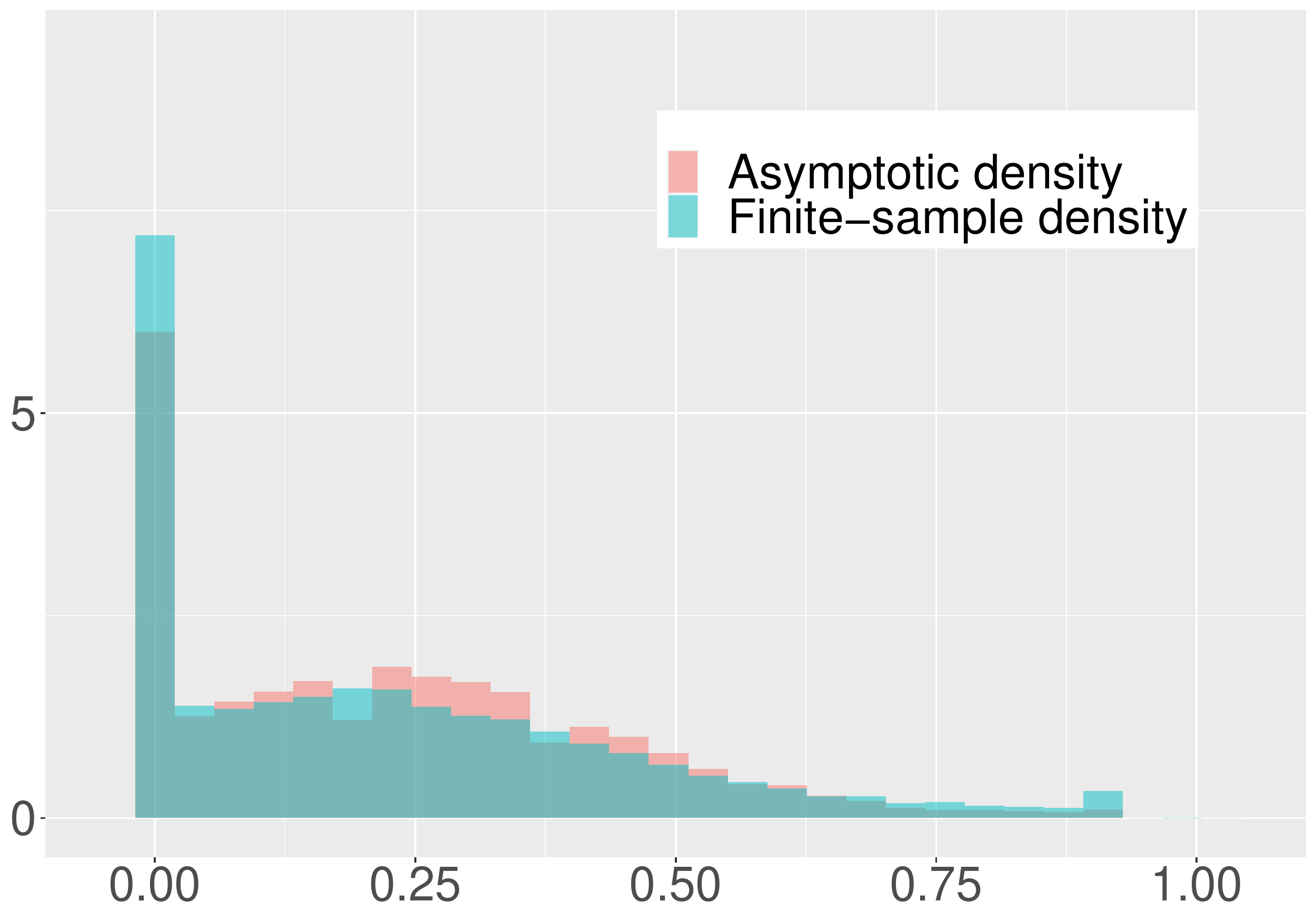}
     \caption{Asymptotic and finite-sample ($n = 500$) densities of $\hat{\pi}_{n}$ with $\sqrt{n} \beta_1 = b_1 = 0, 2,$ and 4 from left to right and $\beta_2 = b_2 = 0$. Here, $\pi = 0.2$, $\varphi = 0.5$, and $\kappa = 0$.}
        \label{plot_pi_n_500}
    \end{center}
\end{figure}

Figure \ref{plot_pi_n_500} shows the asymptotic and finite-sample ($n=500$) densities of $\hat{\pi}_n$ for several values of $b_1$ and $b_2 = 0$; see Section \ref{AsySz} for details on the data generating process ({\textit{dgp}}) and Appendix \ref{CD} for details on how the asymptotic distribution is simulated.\footnote{Figure \ref{plot_beta_zeta_n_500} in Appendix \ref{additional_graphs} plots the corresponding distributions for the elements of $\hat{\psi}_n$.} We observe that the asymptotic distribution provides a good approximation to the finite-sample distribution. Furthermore, we observe that identification strength strongly impacts the distribution of $\hat{\pi}_n$, with large deviations from normality. When $b = 0$, the distribution of $\hat{\pi}_n$ is (almost) completely flat with point masses at the boundaries of the optimization parameter space as well as at 1; recall that $\hat{\pi}_n = 1$ whenever $\hat{\pi}_n = \Pi$. And, as $b_1$ increases, the distribution of $\hat{\pi}_n$ {starts} to resemble the normal distribution (centered at the true value, $\pi = 0.2$, and truncated to the support $\Pi$). This observation is in line with the asymptotic distribution results for $\| b \| = \infty$; see Section \ref{ad_b_infty}. 

The asymptotic distribution result in \eqref{asy_dist_objective} allows us to derive the asymptotic distributions of $LR_n^\dagger$ and $LR_n$ under $\{\gamma_n\} \in \Gamma(\gamma_0,0,b)$ with $\| b \| < \infty$. Let $S_\beta = [I_2 \ 0_2]$ denote the selection matrix that (among $\psi$) selects the entries pertaining to $\beta$ and let $\hat{\lambda}_\beta(\pi;\gamma_0,b) = S_\beta \hat{\lambda}(\pi;\gamma_0,b)$. Then, the asymptotic distribution of $LR^\dagger_n$ under $\{\gamma_n\} \in \Gamma(\gamma_0,0,b)$ with $\| b \| < \infty$ is given by
\begin{equation} \label{asy_dist_LR_star}
	LR^\dagger(\gamma_0,b) =  \max_{\pi \in \Pi} \hat{\lambda}_\beta(\pi;\gamma_0,b)'(c_0S_\beta H^{-1}(\pi;\gamma_0) S_\beta')^{-1}\hat{\lambda}_\beta(\pi;\gamma_0,b) .
\end{equation}
We note that for $b = 0$ we recover the asymptotic distribution of $LR^\dagger_n$ under $H_0^\dagger$, see also Theorem 2.1 in PR. The asymptotic distribution of $LR_n$ under $\{\gamma_n\} \in \Gamma(\gamma_0,0,b)$ with $\| b \| < \infty$ is given by\footnote{If $\{z_t\}$ and $\{x_t\}$ are independent, in which case PR's ``additional assumption'' is also satisfied, then $S_\beta H^{-1}(\pi;\gamma_0) S_\beta'$ is diagonal and the asymptotic distributions of $LR_n^\dagger$ and $LR_n$ simplify. Letting $S_{\beta_1} = [1 \ 0 \ 0]$ and $S_{\beta_2} = [0 \ 1 \ 0]$ as well as $\hat{\lambda}_{\beta_1}(\pi;\gamma_0,b) = S_{\beta_1} \hat{\lambda}(\pi;\gamma_0,b)$ and $\hat{\lambda}_{\beta_2}(\pi;\gamma_0,b) = S_{\beta_2} \hat{\lambda}(\pi;\gamma_0,b)$, we have, for $i \in \{1,2\}$,
\[
	\hat{\lambda}_{\beta_i}(\pi;\gamma_0,b) \sim \max(Z_{\beta_i}(\pi;\gamma_0,b),0),
\]
with $Z_{\beta_i}(\pi;\gamma_0,b) = S_{\beta_i} Z(\pi;\gamma_0,b)$. Furthermore, we have $S_{\beta_1} \hat{\lambda}^r(\pi;\gamma_0,b) = S_{\beta_1} \hat{\lambda}(\pi;\gamma_0,b)$. Letting $\Delta(\pi) = (c_0S_\beta H^{-1}(\pi;\gamma_0) S_\beta')^{-1}$, the asymptotic distribution of $LR_n$, for example, is then given by
\[
	\max_{\pi \in \Pi} \left[ \hat{\lambda}_{\beta_1}(\pi;\gamma_0,b)^2 \Delta_{11}(\pi) + \hat{\lambda}_{\beta_2}(\pi;\gamma_0,b)^2 \Delta_{22}(\pi) \right] - \max_{\pi \in \Pi}  \hat{\lambda}_{\beta_1}(\pi;\gamma_0,b)^2 \Delta_{11}(\pi).
\]}
\begin{align}
	LR(\gamma_0,b) =& \max_{\pi \in \Pi} \hat{\lambda}_\beta(\pi;\gamma_0,b)'(c_0S_\beta H^{-1}(\pi;\gamma_0) S_\beta')^{-1}\hat{\lambda}_\beta(\pi;\gamma_0,b) \nonumber \\ -& \max_{\pi \in \Pi} \hat{\lambda}^r_\beta(\pi;\gamma_0,b)'(c_0S_\beta H^{-1}(\pi;\gamma_0) S_\beta')^{-1}\hat{\lambda}^r_\beta(\pi;\gamma_0,b), \label{asy_dist_LR}
\end{align}
where $\hat{\lambda}^r_\beta(\pi;\gamma_0,b) = S_{\beta} \hat{\lambda}^r(\pi;\gamma_0,b) $,
\[
\hat{\lambda}^r(\pi;\gamma_0,b) = \argmin_{\lambda \in \Lambda^r} q(\lambda,\pi;\gamma_0,b),
\]
and 
\[
	 \Lambda^r = [0,\infty] \times \{ 0 \} \times [-\infty,\infty].
\]
The results in \eqref{asy_dist_LR_star} and \eqref{asy_dist_LR} provide us with the asymptotic distributions of $LR^\dagger_n$ and $LR_n$, respectively, for true values of $\beta_1$ that are close to the boundary relative to the sample size, $b_1 < \infty$, which are permitted under $H_0$. Together, these results allow us to determine the asymptotic null rejection frequencies of the testing procedures proposed by PR under empirically relevant values of $\beta_1$. These asymptotic null rejection frequencies play a crucial role in determining whether the testing procedures control asymptotic size; see Section \ref{AsySz}. 

\begin{figure}[h!] 
  \begin{center}
        \includegraphics[width=54mm]{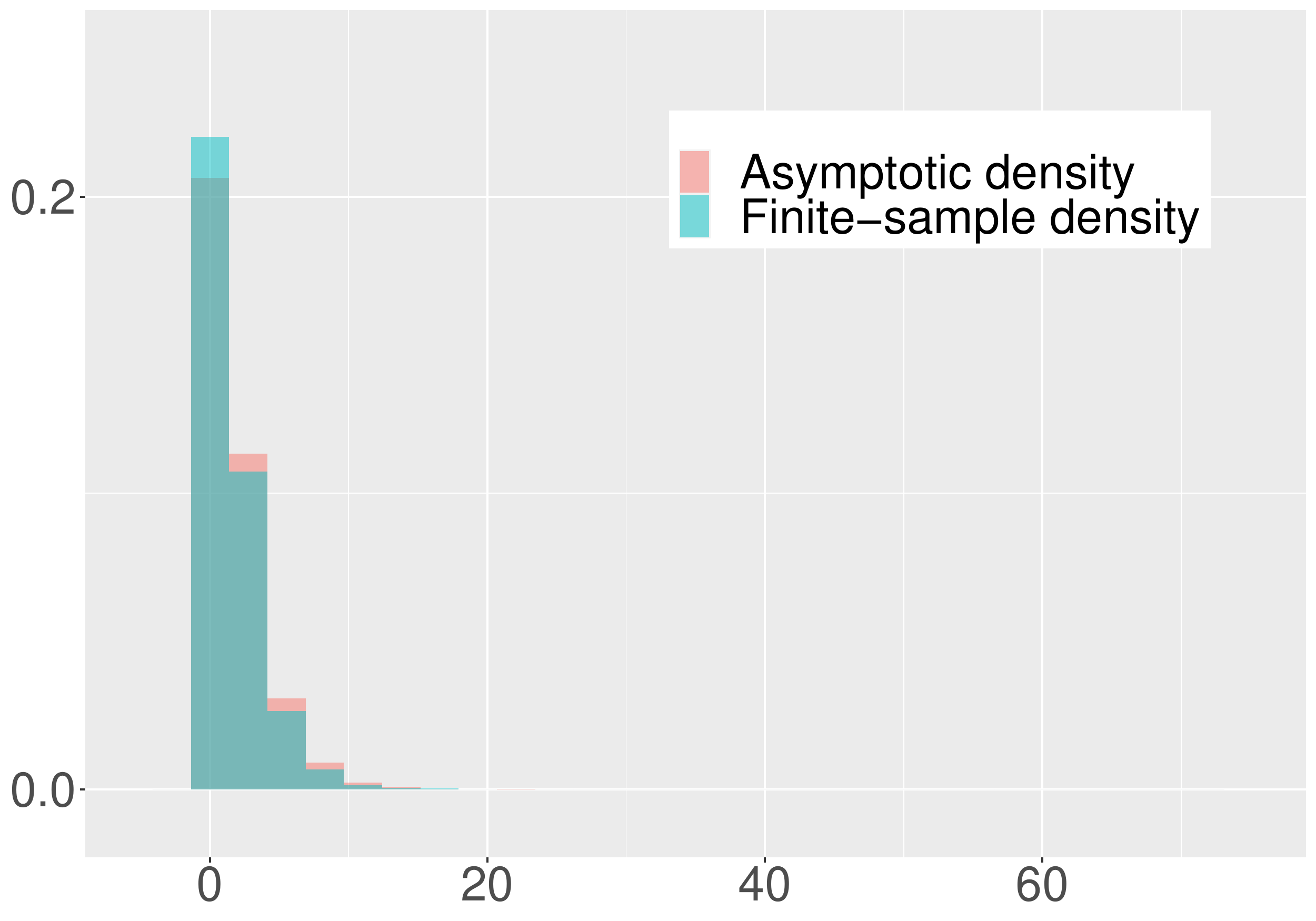}
    \includegraphics[width=54mm]{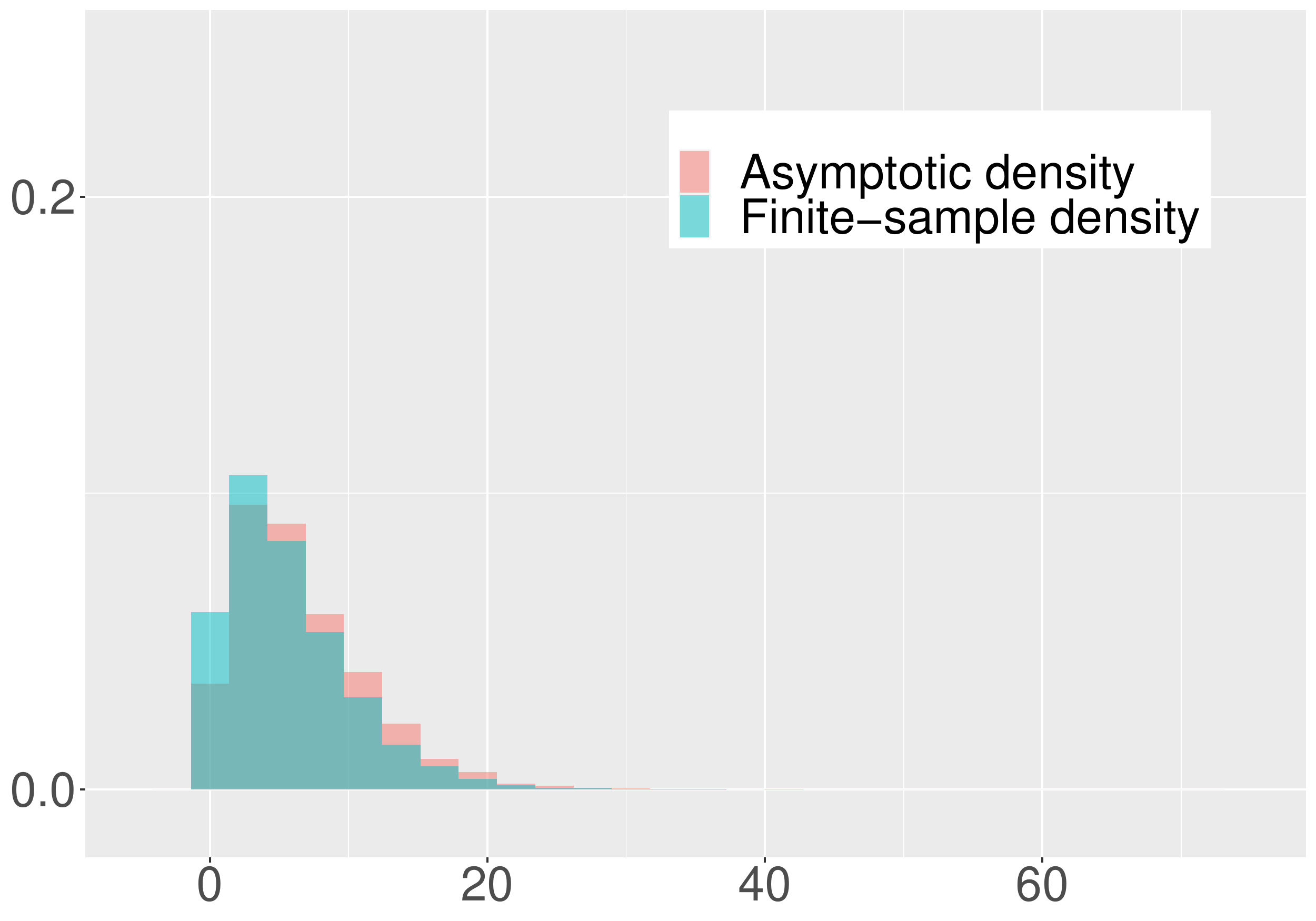}
    \includegraphics[width=54mm]{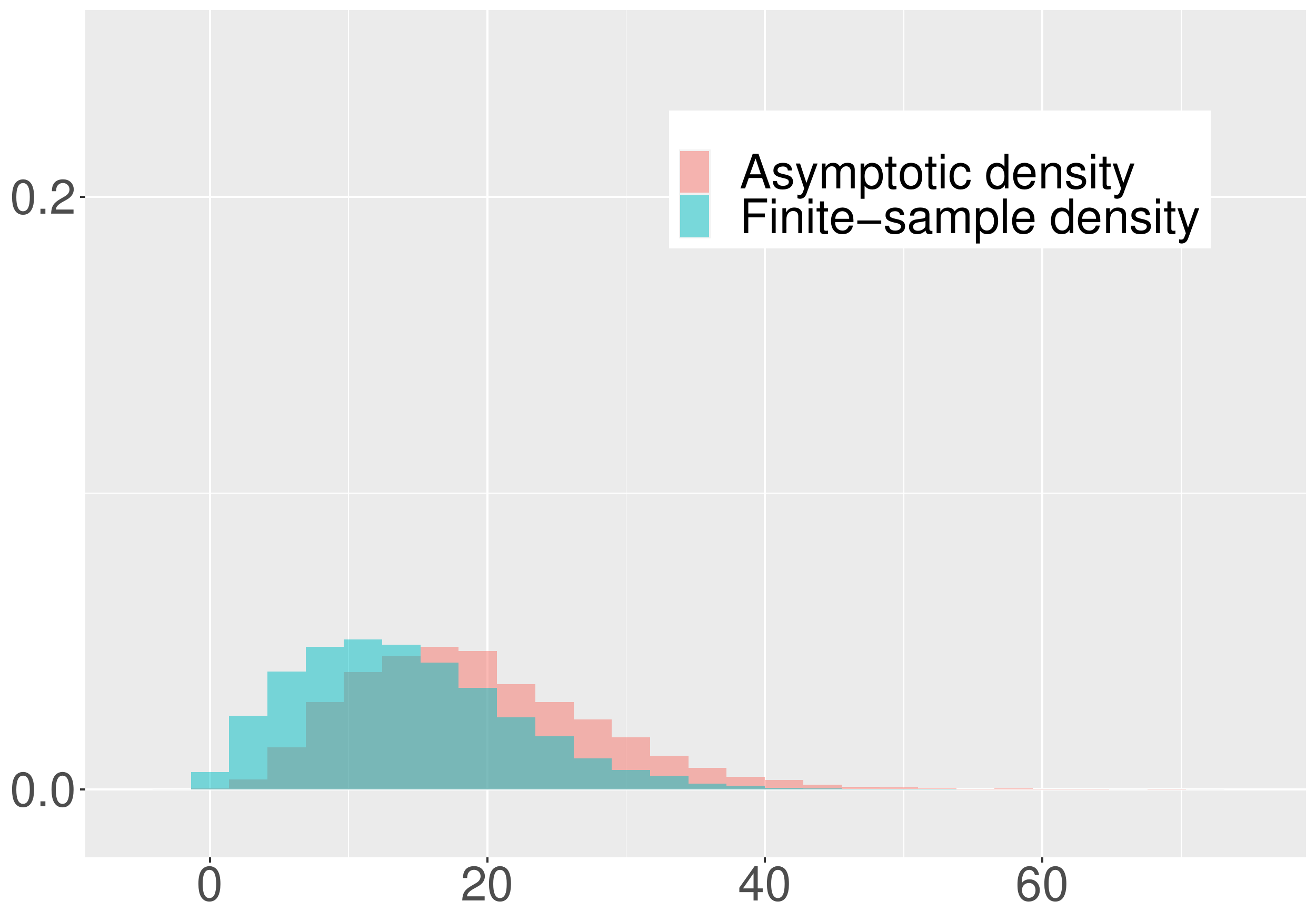}
        \includegraphics[width=54mm]{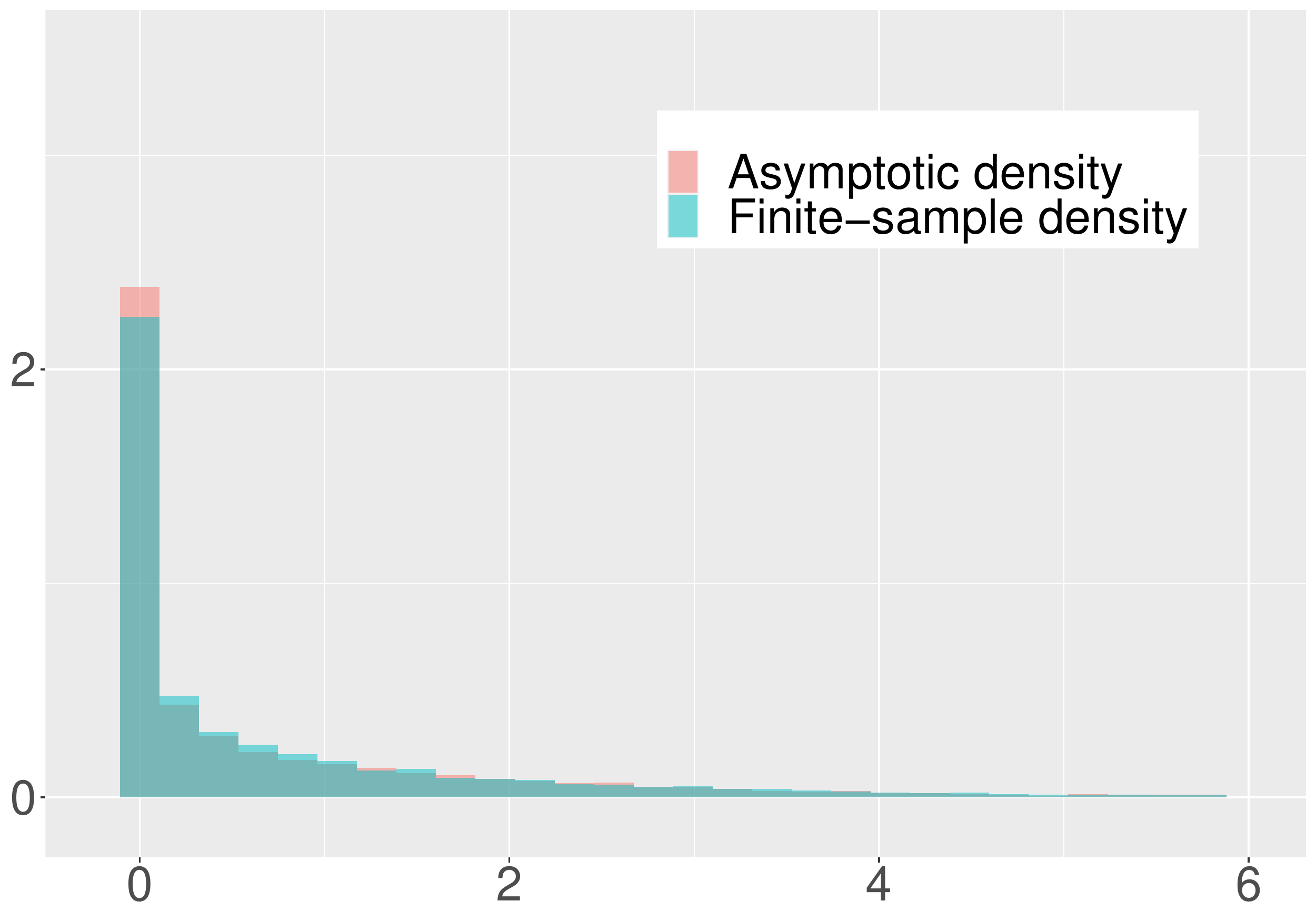}
    \includegraphics[width=54mm]{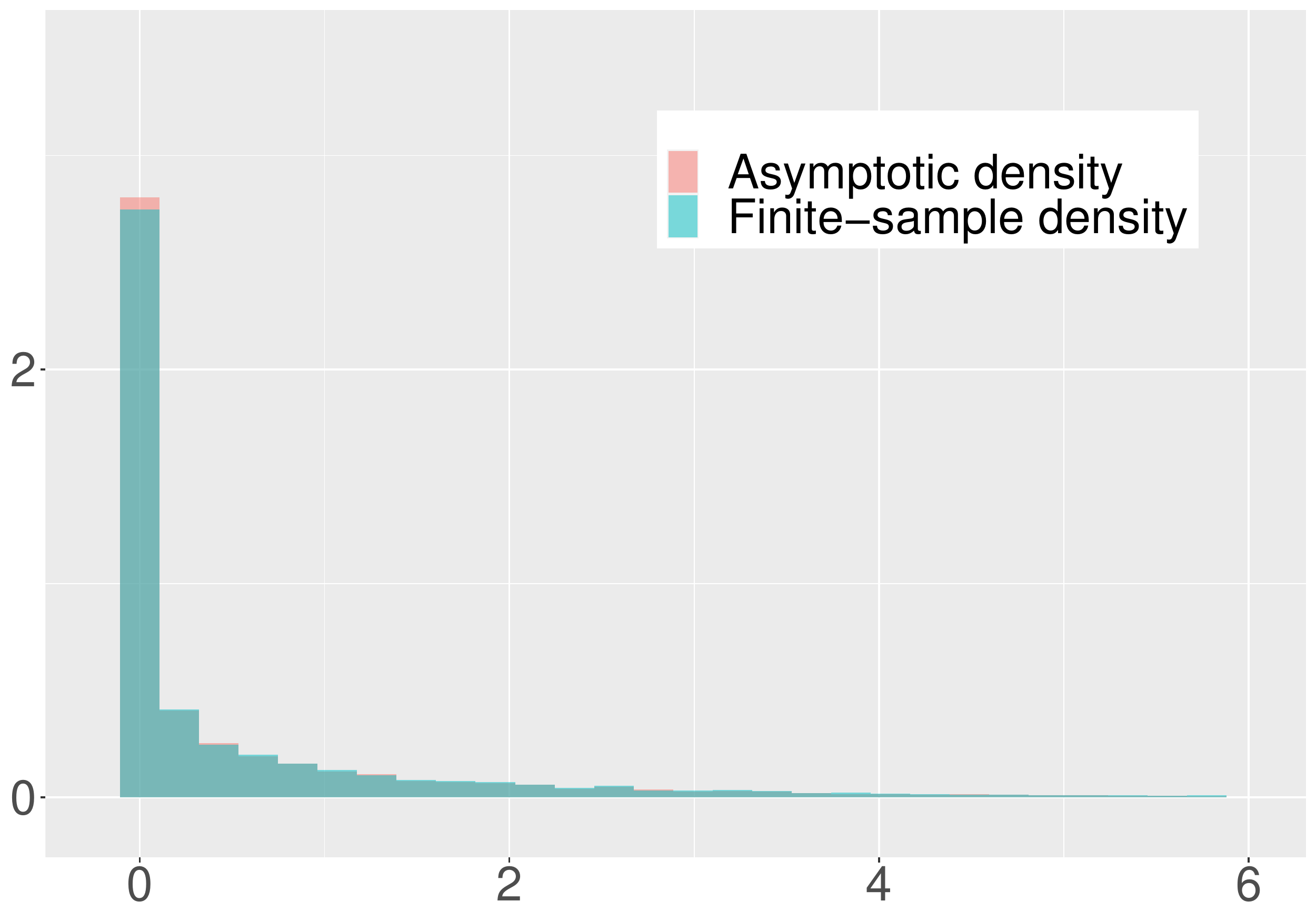}
    \includegraphics[width=54mm]{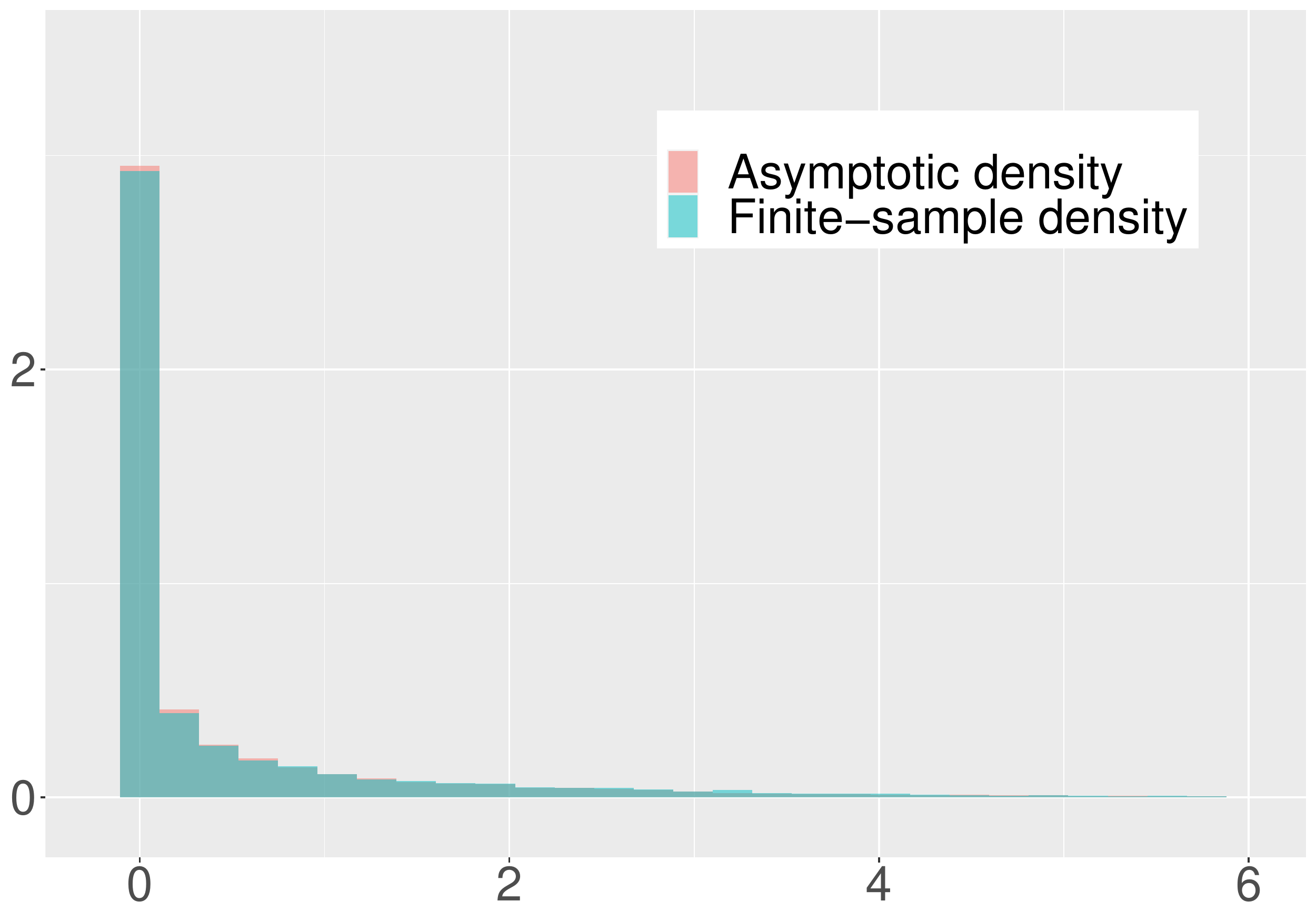}
     \caption{Asymptotic and finite-sample ($n = 500$) densities of $LR_n^\dagger$ (top row) and $LR_n$ (bottom row) with $\sqrt{n} \beta_1 = b_1 = 0,2,$ and 4 from left to right and $\beta_2 = b_2 = 0$. Here, $\pi = 0.2$, $\varphi = 0.5$, and $\kappa = 0$.}
       \label{plot_LRs_n_500}
    \end{center}
\end{figure}

Figure \ref{plot_LRs_n_500} shows the asymptotic and the finite-sample ($n = 500$) densities of $LR_n^\dagger$  and $LR_n$ for the same \textit{dgp} that underlies Figure \ref{plot_pi_n_500}. While the asymptotic distribution of $LR_n$ provides a very good approximation to its finite-sample distribution, the asymptotic distribution of $LR_n^\dagger$ seems to be first-order stochastically dominated by its finite-sample distribution. However, looking at Figure \ref{plot_LRs_n_10000}, which reproduces the top row of Figure \ref{plot_LRs_n_500} with $n=$ 10,000, this seems to be a ``small sample'' phenomenon. 

\begin{figure}[h!] 
  \begin{center}
        \includegraphics[width=54mm]{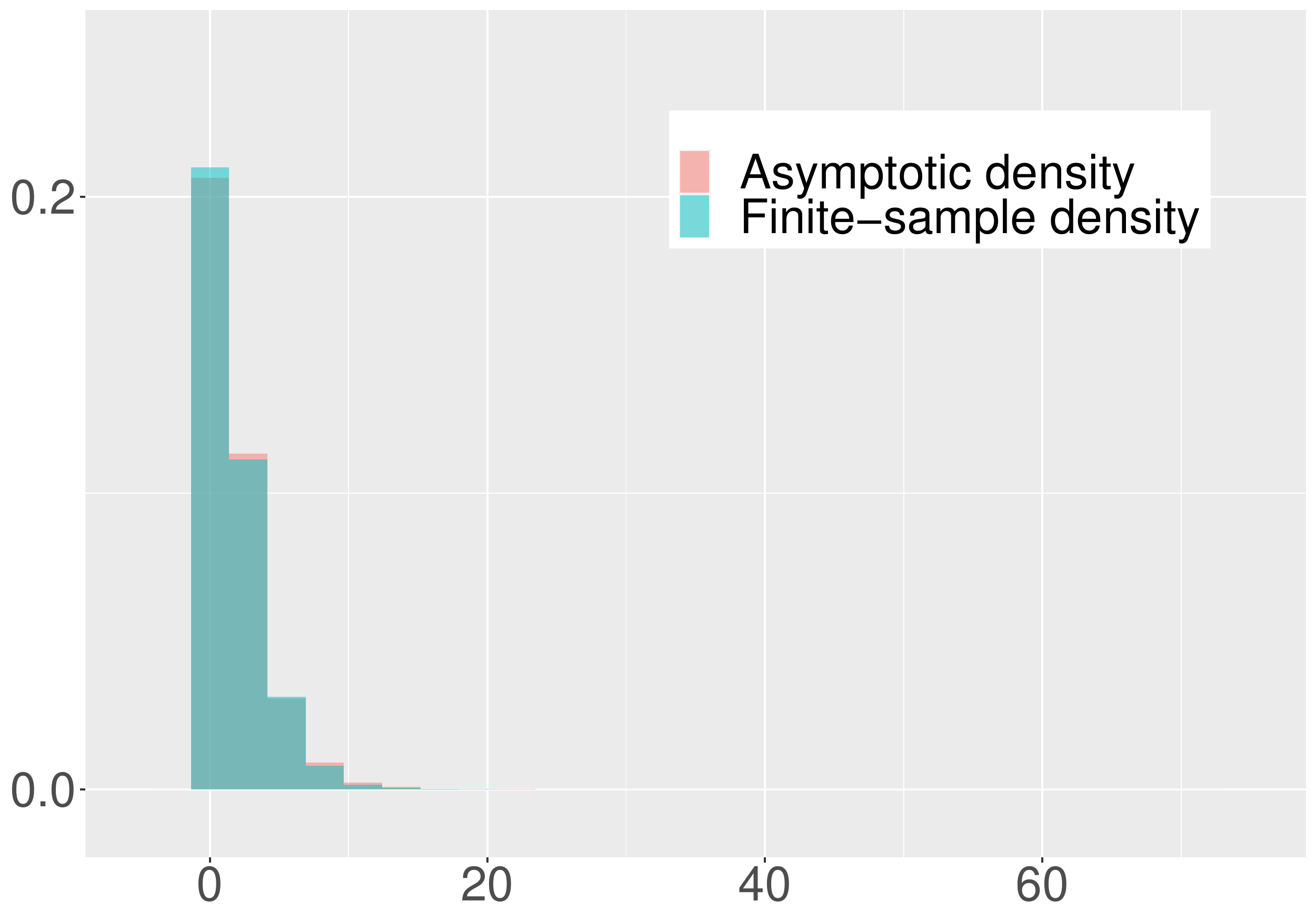}
    \includegraphics[width=54mm]{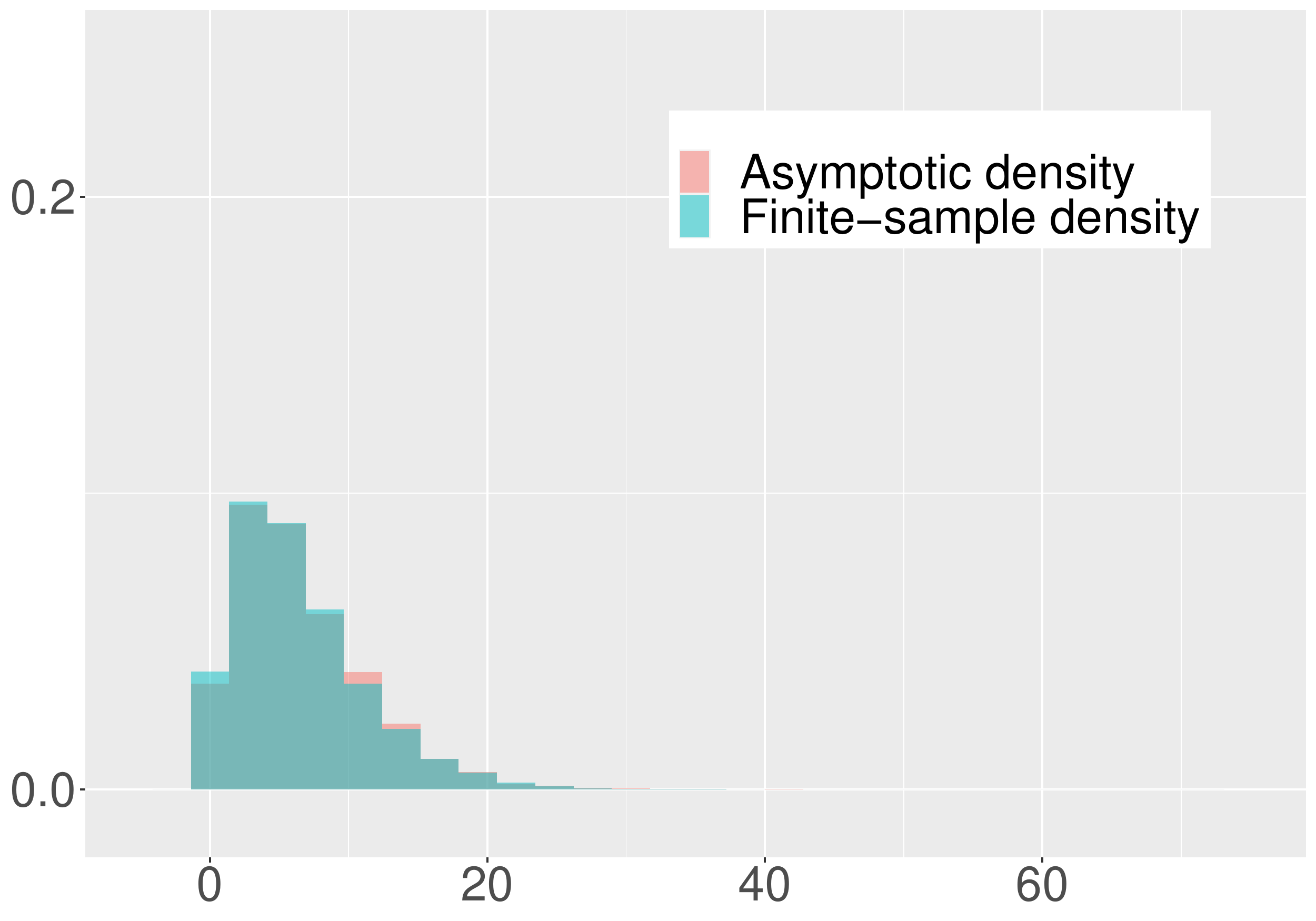}
    \includegraphics[width=54mm]{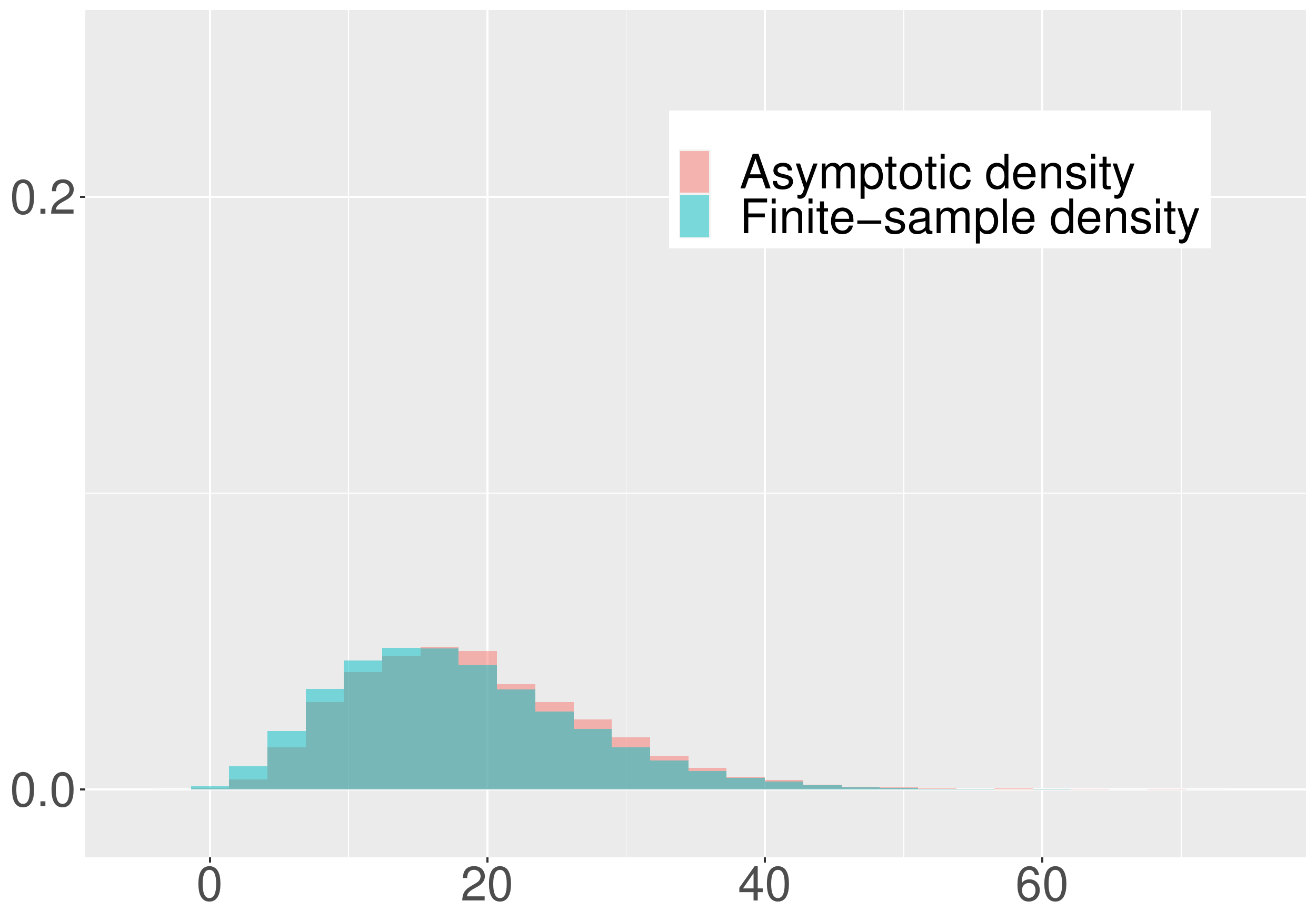}
     \caption{Asymptotic and finite-sample ($n =$ 10,000) densities of $LR_n^\dagger$ with $\sqrt{n} \beta_1 = b_1 = 0, 2,$ and 4 from left to right and $\beta_2 = b_2 = 0$. Here, $\pi = 0.2$, $\varphi = 0.5$, and $\kappa = 0$.}
       \label{plot_LRs_n_10000}
    \end{center}
\end{figure}

\subsubsection{Results for $\| b \| = \infty$ (and $\beta_n \to 0$)} \label{results_b_infty}

Under $\{\gamma_n\} \in \Gamma(\gamma_0,\infty,b,\omega_0,p)$ with $\beta_0 = 0$, we have that
\begin{equation} \label{Lemma32b}
	\| \beta_n \|^{-2} (Q_n^c(\pi) - Q_{0,n}) \to \eta(\pi;\gamma_0,\omega_0) 
\end{equation}
uniformly over $\pi \in \Pi$, where
\[
	\eta(\pi;\gamma_0,\omega_0) = -\frac{1}{2} \omega_0'K(\pi;\gamma_0)' H^{-1}(\pi;\gamma_0) K(\pi;\gamma_0) \omega_0.
\]
Furthermore, we have that $\eta(\pi;\gamma_0,\omega_0)$ is uniquely minimized at $\pi = \pi_0 \ \forall \gamma_0 \in \Gamma$ with $\beta_0 = 0$. It can then be shown that $\hat{\pi}_n - \pi_n \to_p 0$ under $\{\gamma_n\} \in \Gamma(\gamma_0,\infty,b,\omega_0,p)$ and $\| \beta_n \|^{-1} (\hat{\psi}_n - \psi_n) = o_p(1)$ under $\{\gamma_n\} \in \Gamma(\gamma_0,\infty,b,\omega_0,p)$ with $\beta_0 = 0$, where the latter result ensures that \eqref{quadratic_expansion_full_vector} holds; see the discussion below Assumption D1 in AC.
 
\subsection{Asymptotic distribution for $\| b \| = \infty$} \label{ad_b_infty}

Under $\{\gamma_n\} \in \Gamma(\gamma_0,\infty,b,\omega_0,p)$, we have 
\begin{equation} \label{quadratic_expansion_full_vector}
	Q_n(\theta) = Q_n(\theta_n) + DQ_n(\theta_n)'(\theta-\theta_n) + \frac{1}{2} (\theta-\theta_n)'D^2Q_n(\theta_n)(\theta - \theta_n) + R^*_n(\theta),
\end{equation}
where the remainder $R^*_n(\theta)$ satisfies
\[
	\sup_{\theta \in \Theta_n(\delta_n)} |nR_n^*(\theta)| = o_p(1)
\]
for all $\delta_n \to 0$, where $\Theta_n(\delta_n) = \{ \theta \in \Theta: \| \psi - \psi_n \| \leq \delta_n \| \beta_n \| \text{ and } \| \pi - \pi_n \| \leq \delta_n\}$. Here, $DQ_n(\theta)$ and $D^2Q_n(\theta)$ denote the vector of first-order left/right partial derivatives and the matrix of second-order left/right partial derivaties of $Q^\infty_n(\theta)$ with respect to $\theta$, respectively. Define 
\[
	B(\beta) = \left[ \begin{array}{cc} I_3 & 0_{3\times1} \\ 0_{1\times3} & \| \beta \| \end{array} \right],
\]
\[
	J_n = B^{-1}(\beta_n) D^2 Q_n(\theta_n) B^{-1}(\beta_n),
\]
and
\[
	Z_n^\infty = - \sqrt{n} J_n^{-1} B^{-1}(\beta_n) DQ_n(\theta_n).
\]
Then \eqref{quadratic_expansion_full_vector} can be rewritten as
\[
	Q_n(\theta) = Q_n(\theta_n) - \frac{1}{2n} {Z_n^\infty}'J_nZ_n^\infty + \frac{1}{2n} q_n^{\infty}(\sqrt{n}B(\beta_n)(\theta-\theta_n)) + R^*_n(\theta),
\]
where
\[
	q_n^{\infty}(\lambda) = (\lambda - Z_n^\infty)'J_n(\lambda - Z_n^\infty).
\]
Note that
 \[
 	B^{-1}(\beta_n)DQ_n(\theta_n) = \frac{1}{n} \sum_{t=1}^n \frac{1}{2h^\infty_t(\theta_n)} \left( 1-  \frac{y_t^2}{h^\infty_t(\theta_n)} \right)  \tau(\beta_n/\|\beta_n\|,\pi_n),
 \]
 where
 \begin{align*}
 	 \tau(\beta/\|\beta\|,\pi) &= B^{-1}(\beta) \frac{\partial h^\infty_t(\theta)}{\partial \theta} \\  & = \left( \sum_{i=0}^{\infty} \pi^i y_{t-i-1}^2, \sum_{i=0}^{\infty} \pi^i x_{t-i-1}^2, 1, \sum_{i=1}^{\infty} i \pi^{i-1} \left(\frac{\beta_{1}}{\|\beta\|} y_{t-i-1}^2 + \frac{\beta_{2}}{\|\beta\|} x_{t-i-1}^2\right) \right)'.
 \end{align*}
 Then, under $\{\gamma_n\} \in \Gamma(\gamma_0,\infty,b,\omega_0,p)$, we have $\sqrt{n} B^{-1}(\beta_n)DQ_n(\theta_n) \overset{d}{\to} N(0,V(\gamma_0,\omega_0))$, where
 \[
 	V(\gamma_0,\omega_0) = \frac{c_0}{2} E_{\gamma_0} \frac{\tau(\omega_0,\pi_0)}{h^\infty_t(\theta_0)}  \frac{\tau'(\omega_0,\pi_0)}{h^\infty_t(\theta_0)},
 \]	
 and $J_n \to_p J(\gamma_0,\omega_0)$, where $J(\gamma_0,\omega_0) = V(\gamma_0,\omega_0)/c_0$, such that
 \begin{equation} \label{Z_infty}
 	Z_n^\infty \overset{d}{\to} Z^\infty(\gamma_0,\omega_0) \sim N(0,c_0J^{-1}(\gamma_0,\omega_0)).
 \end{equation}
 Let 
 \[
 	q^{\infty}(\lambda;\gamma_0,\omega_0) = (\lambda - Z^\infty(\gamma_0,\omega_0))'J(\gamma_0,\omega_0)(\lambda - Z^\infty(\gamma_0,\omega_0)).
 \]
 Then, under $\{\gamma_n\} \in \Gamma(\gamma_0,\infty,b,\omega_0,p)$, we have
 \begin{equation} \label{lambda_infty}
 	\sqrt{n}B(\beta_n)(\hat{\theta}_n - \theta_n) \overset{d}{\to} \hat{\lambda}^\infty(\gamma_0,\omega_0,b,p),
 \end{equation}
 where $\hat{\lambda}^\infty(\gamma_0,\omega_0,b,p) = \argmin_{\lambda \in \Lambda^\infty} q^{\infty}(\lambda;\gamma_0,\omega_0)$ and where
 \[
 	\Lambda^\infty = [-b_1,\infty] \times  [-b_2,\infty] \times [-\infty,\infty] \times [-p,\infty].
 \]
 Furthermore, under $\{\gamma_n\} \in \Gamma(\gamma_0,\infty,b,\omega_0,p)$, we have
 \begin{equation} \label{obj_fun_infty}
	2n(Q_n(\hat{\theta}_n) - Q_{n}(\theta_n)) \overset{d}{\to} - \hat{\lambda}^\infty(\gamma_0,\omega_0,b,p)'J(\gamma_0,\omega_0)\hat{\lambda}^\infty(\gamma_0,\omega_0,b,p).
\end{equation}
Let $S_{\beta,\pi} = [I_2 \ 0_{2\times2}; 0_3' \ 1]$ be the selection matrix that selects (among $\theta$) the entries pertaining to $\beta$ and $\pi$ and define $\hat{\lambda}^\infty_{\beta,\pi}(\gamma_0,\omega_0,b,p) = S_{\beta,\pi} \hat{\lambda}^\infty(\gamma_0,\omega_0,b,p)$. Also, let 
\[
	\lambda^{\infty,r}(\gamma_0,\omega_0,b,p) = \argmin_{\lambda \in \Lambda^{\infty,r}} q^{\infty}(\lambda;\gamma_0,\omega_0),
\]
where 
 \[
 	\Lambda^{\infty,r} = [-b_1,\infty] \times  [-b_2,-b_2] \times [-\infty,\infty] \times [-p,\infty],
 \]
 and define $\hat{\lambda}^{\infty,r}_{\beta,\pi}(\gamma_0,\omega_0,b,p) = S_{\beta,\pi} \hat{\lambda}^{\infty,r}(\gamma_0,\omega_0,b,p)$. Then, the asymptotic distribution of $LR_n$, under $\{\gamma_n\} \in \Gamma(\gamma_0,\infty,b,\omega_0,p)$, is given by
\begin{align}
	LR^\infty(\gamma_0,\omega_0,b,p) =& \hat{\lambda}^\infty_{\beta,\pi}(\gamma_0,\omega_0,b,p)'(cS_{\beta,\pi} J^{-1}(\gamma_0,\omega_0) S_{\beta,\pi}')^{-1}\hat{\lambda}^\infty_{\beta,\pi}(\gamma_0,\omega_0,b,p) \nonumber \\ -& \hat{\lambda}^{\infty,r}_{\beta,\pi}(\gamma_0,\omega_0,b,p)'(cS_{\beta,\pi} J^{-1}(\gamma_0,\omega_0) S_{\beta,\pi}')^{-1}\hat{\lambda}^{\infty,r}_{\beta,\pi}(\gamma_0,\omega_0,b,p). \label{LR_infty}
\end{align}

In the following section, we use the foregoing asymptotic distribution results to analyze the asymptotic size of the two testing procedures proposed by PR.
 
 \section{Asymptotic size and a new test} \label{AsySz}
 
First, we provide a characterization result for asymptotic size that holds for any test for testing \eqref{testing_problem}. To that end, let
\[
	H = \{ h = (b_1, \gamma_0): 0 \leq b_1 < \infty \text{ and }   \gamma_0 \in \text{cl}(\Gamma) \text{ with } \beta_0 = 0 \}
\]
and
\[
	 H^\infty =   \{ h =  (p , \gamma_0) : 0 \leq p < \infty \text{ and } \gamma_0 \in \text{cl}(\Gamma) \text{ with } \pi_0 = 0 \}.
\]

Furthermore, let $RP_\mathcal{T}(h)$ denote the asymptotic rejection probability of the test (that uses $\mathcal{T}_n$) under $\{\gamma_n\} \in \Gamma(\gamma_0,0,b)$ with $b_1 < \infty$ and $b_2 = 0$, where $h = (b_1,\gamma_0)$. Similarly, let $RP_\mathcal{T}^\infty(h)$ denote the asymptotic rejection probability of the test under $\{\gamma_n\} \in \Gamma(\gamma_0,\infty,b,\omega_0,p)$ with $\pi_0 = 0$, $b_1 = \infty$, $b_2 = 0$, and $\omega_0 = e_1$, where $h = (p,\gamma_0)$. Lastly, let $RP_\mathcal{T}^{\infty,\infty} \in [0,1]$ be such that $\limsup_{n \to \infty} RP_{\mathcal{T},n}(\gamma_n) \geq RP_\mathcal{T}^{\infty,\infty}$ for any $\{\gamma_n\} \in \Gamma(\gamma_0,\infty,b,\omega_0,\infty)$ with $\pi_0 = 0$, $b_1 = \infty$, $b_2 = 0$, and $\omega_0 = e_1$ and $RP_{\mathcal{T},n}(\gamma_n) \to RP_\mathcal{T}^{\infty,\infty}$ for some $\{\gamma_n\} \in \Gamma(\gamma_0,\infty,b,\omega_0,\infty)$ with $\pi_0 = 0$, $b_1 = \infty$, $b_2 = 0$, and $\omega_0 = e_1$, where $RP_{\mathcal{T},n}(\gamma_n)$ denotes the finite-sample rejection probability of the test under $\gamma_n$. We can now state our characterization result for AsySz$_\mathcal{T}$.

\begin{pro} \label{pro1} Given $\Gamma$ as defined in Appendix \ref{ACverification}, we have 
\[
	\textup{AsySz}_\mathcal{T} = \max \{\sup_{h \in H} RP_\mathcal{T}(h), \sup_{h \in H^\infty}  RP_\mathcal{T}^\infty(h) , RP_\mathcal{T}^{\infty,\infty} \}
\]
for any test for testing \eqref{testing_problem}.
\end{pro}

The proof follows along the lines of the proof of Lemma 2.1 in AC and is given in Appendix \ref{ACverification}, along with the verification of all claims made in this section.

Next, we consider the two testing procedures proposed by PR. We have
\[
	RP_{TS}(h) = P\left( \mathbbm{1}(LR^\dagger(\gamma_0,b) > LR^\dagger_{1-\alpha}(\gamma_0,0)) \times LR(\gamma_0,b) > \text{cv}_{1-\alpha}  \right),
\]
where $LR^\dagger_{1-\alpha}(\gamma_0,0)$ denotes the $1-\alpha$ quantile of $LR^\dagger(\gamma_0,0)$, and
\[
	RP_{S}(h) = P\left( LR(\gamma_0,b) > \text{cv}_{1-\alpha}  \right).
\]
We note that $RP_{TS}(h)$ and $RP_{S}(h)$ depend on $\gamma_0$ through $H(\pi;\gamma_0)$, $\Omega(\pi_1, \pi_2;\gamma_0)$, and $K(\pi;\gamma_0)$.

Since $LR^\dagger_n \to \infty$ and $P(\widetilde{LR}^\dagger_{n,1-\alpha} < \infty) = 1$ under $\{\gamma_n\} \in \Gamma(\gamma_0,\infty,b,\omega_0,p)$ with $\pi_0 = 0$, $b_1 = \infty$, $b_2 = 0$, and $\omega_0 = e_1$, we have that
\[
	RP_{TS}^\infty(h) = RP_{S}^\infty(h) =  P\left( LR^\infty(\gamma_0,p) > \text{cv}_{1-\alpha}  \right),
\]
where $LR^\infty(\gamma_0,p) = LR^\infty(\gamma_0,e_1,(\infty,0)',p)$ with $\pi_0 = 0$. Now, let  
\[
 \rho(\gamma_0) = \frac{\left( J^{-1}(\gamma_0,e_1)\right)_{\beta_2,\pi}}{\sqrt{\left( J^{-1}(\gamma_0,e_1)\right)_{\beta_2,\beta_2}\left( J^{-1}(\gamma_0,e_1)\right)_{\pi,\pi}}}
\]
and $q(\gamma_0) = q(\rho(\gamma_0)) = \sin^{-1}(\rho(\gamma_0))/2\tilde{\pi}$. Furthermore, let ${LR}^\infty_{1-\alpha}(\gamma_0,0)$ denote the $1-\alpha$ quantile of ${LR}^\infty(\gamma_0,p)$ and note that $LR_{1-\alpha}^\infty(\gamma_0,\infty) = \text{cv}_{1-\alpha}$. It can be shown that $LR_{1-\alpha}^\infty(\gamma_0,p)$ is strictly monotonically increasing (decreasing) in $p$ if $ \rho(\gamma_0) < 0$ ($ \rho(\gamma_0) > 0$), while $LR_{1-\alpha}^\infty(\gamma_0,p) = \text{cv}_{1-\alpha}$ if $\rho(\gamma_0) = 0$; this result is in line with and extends the result in PR. It then follows, from Theorem 2.1 in \cite{Kopylev:11}, that 
\begin{align}
	& \sup_{h \in H^\infty} P\left( LR^\infty(\gamma_0,b) > \text{cv}_{1-\alpha}  \right) \nonumber \\
	=& \max \{ \sup_{\gamma_0:(0,\gamma_0) \in H^\infty, \rho(\gamma_0) > 0} 1 - ((1/2 - q(\gamma_0)) + 1/2 F_{\chi^2}(\text{cv}_{1-\alpha};1) + q(\gamma_0) F_{\chi^2}(\text{cv}_{1-\alpha};2)), \alpha \}, \label{cf_expression}
\end{align}
where $F_{\chi^2}(\cdot;k)$ denotes the cumulative distribution function of a $\chi^2$ random variable with degree of freedom $k$. Noting that $ 1 - ((1/2 - q(\gamma_0)) + 1/2 F_{\chi^2}(\text{cv}_{1-\alpha};1) + q(\gamma_0) F_{\chi^2}(\text{cv}_{1-\alpha};2))$ is strictly increasing in $\rho(\gamma_0)$, we have that $\sup_{h \in H^\infty} P\left( LR^\infty(\gamma_0,b) > \text{cv}_{1-\alpha}  \right)$ is bounded from above by $11.46\%$ for $\alpha = 0.05$; this number is obtained by evaluating \eqref{cf_expression} at $\rho(\gamma_0) = 1$. 

In what follows, we determine lower bounds on $\sup_{h \in H} RP_{TS}(h)$ and $\sup_{h \in H} RP_{S}(h)$ by numerically evaluating $RP_{TS}(h)$ and $RP_{S}(h)$ at certain choices of $h \in H$. Similarly, we determine a lower bound on \eqref{cf_expression} by numerically evaluating $\rho(\gamma_0)$ at certain choices of $\gamma_0$ such that $(0,\gamma_0) \in H^\infty$ (and by plugging the resulting value into $1 - ((1/2 - q(\gamma_0)) + 1/2 F_{\chi^2}(\text{cv}_{1-\alpha};1) + q(\gamma_0) F_{\chi^2}(\text{cv}_{1-\alpha};2))$). Combining the above lower bounds then provides us with lower bounds on AsySz$_{TS}$ and AsySz$_{S}$.

The numerical evaluation is based on the following \textit{dgp}, which is inspired by PR (see their Appendix D). $\phi$ is such that $x_t$ follows an AR(1), i.e., 
\begin{equation} \label{phi}
	x_t = \varphi x_{t-1} + \epsilon_t,
\end{equation}
where
\[
	\left( \begin{array}{c} z_t \\ \epsilon_t \end{array} \right) \sim N\left( \left( \begin{array}{c} z_t \\ \epsilon_t \end{array} \right), \left( \begin{array}{cc} 1 & \kappa \\ \kappa & 1 \end{array} \right) \right).
\] 
We note that, given the above $\phi$, the ``additional assumption'' of PR is satisfied if and only if $\kappa = 0$. In what follows, we take $\bar{\pi} = 0.9$ and obtain results for $\alpha = 5\%$; details on the numerical evaluation can be found in Appendix \ref{CD}. We find that $RP_{TS}(h) = 6\%$ where $h$ is given by $b_1 = 2.5$, $\zeta_0 = 1$, $\pi_0 = 0.64$, $\varphi_0 = 0.5$, and $\kappa_0 = 0$. Furthermore, we find that $RP_{S}(h) = 9.48\%$ where $h$ is as before, except that $b_1 = 0$ (and $\pi_0 = \Pi^*$). Lastly, we find that $RP^\infty_{TS}(h) = RP^\infty_{S}(h) = 6.65\%$ where $h$ is given by $p = 0$, $\beta_1 = 0.3$, $\zeta_0 = 1$, $\pi_0 = 0$, $\varphi_0 = 0$, and $\kappa_0 = 0.99$ so that $\rho(\gamma_0) = 0.39$.\footnote{This lower bound is increasing in $\kappa_0$, e.g., we have $RP^\infty_{TS}(h) = RP^\infty_{S}(h) = 7.11\%$ where $h$ is as before except that $\beta_1 = 0.25$ and $\kappa_0 = 0.999$ so that $\rho(\gamma_0) = 0.49$. We refrain from reporting results for larger $\kappa_0$ due to the decreasing accuracy in the numerical evaluation (as $\kappa_0$ approaches 1).} Noting that $RP_{TS}^{\infty,\infty} = RP_{S}^{\infty,\infty} = \alpha$, we obtain the following Corollary to Proposition \ref{pro1}.

\begin{cor} \label{cor1}
Given $\Gamma$ as defined in Appendix \ref{ACverification} with $\overline{\pi} = 0.9$, AsySz$_{TS} \geq 6.65\%$ and AsySz$_{S} \geq 9.48\%$ at $\alpha = 5\%$. 
\end{cor}

\begin{rem} \label{}
We note that the two lower bounds obtained in Corollary \ref{cor1} also apply if $\Gamma$ is restricted to satisfy the ``additional assumption'' in PR and, in case of the second testing procedure, the assumption that $\beta_1 > 0$. This is due to the point-wise nature of the assumptions and the ``suprema'' in the definitions of Sz$_\mathcal{T}$ and AsySz$_\mathcal{T}$. For example, the ``additional assumption'' in PR is only imposed at $\pi = 0$ and does not exclude sequences of true parameters that are such that $\rho(\gamma_n) \to \rho(\gamma_0) > 0$ and $\pi_n > 0$ for all $n \geq 1$, while $\pi_n \to \pi_0 = 0$.
\end{rem}

Figure \ref{plot_LRs_n_500} shows that the (asymptotic) distribution of $LR_n$ does not vary a lot with $\beta_1$, at least for the particular $\gamma$ under consideration. This motivates our suggestion to test \eqref{testing_problem} using $LR_n$ combined with a plug-in least favorable configuration critical value (PI-LF); in what follows, the resulting test is abbreviated as $LR$-$LF$. In the context at hand, we have to consider two identification scenarios ($b_1 < \infty$ and $b_1 = \infty$) that result in two different asymptotic null distributions of $LR_n$. The idea is simple: in each scenario, all unknown quantities/parameters of the asymptotic distribution of $LR_n$ that are consistently estimable are replaced by estimators. For the remaining parameters, we determine the least favorable configuration, i.e., we determine under which values of these parameters the $1-\alpha$ quantile of the asymptotic distribution is maximized. The corresponding $1-\alpha$ quantile then serves as critical value for the scenario under consideration. Finally, the maximum of the thus obtained critical values and $\text{cv}_{1-\alpha}$ constitutes our proposed PI-LF. We note that, under $\{\gamma_n\} \in \Gamma(\gamma_0,0,b)$ with $b_1 < \infty$ and $b_2 = 0$, $b_1$ and $\pi_0$ are not consistently estimable, while, under $\{\gamma_n\} \in \Gamma(\gamma_0,\infty,b,\omega_0,p)$ with $\pi_0 = 0$, $b_1 = \infty$, $b_2 = 0$, and $\omega_0 = e_1$, $p$ is not consistently estimable. Our final implementation follows the above (``general'') approach with two twists. First, under $\{\gamma_n\} \in \Gamma(\gamma_0,0,b)$ with $b_1 < \infty$ and $b_2 = 0$, we restrict the search of the least favorable configuration to values of $b_1$ and $\pi_0$ that satisfy $3(\sqrt{n}b_1)^2+\pi_0^2+2(\sqrt{n}b_1)\pi_0 < 1$; this condition is motivated by the condition $E_{\gamma} y_t^2 < \infty$ \citep[see e.g.,][]{MS:00} that is also imposed by $\Gamma$. Second, if $\hat{\beta}_{n,1} = 0$, then we only consider the critical value obtained from the $b_1 < \infty$ scenario, because the probability of observing $\hat{\beta}_{n,1} = 0$ under $\{\gamma_n\} \in \Gamma(\gamma_0,\infty,b,\omega_0,p)$ with $\pi_0 = 0$, $b_1 = \infty$, $b_2 = 0$, and $\omega_0 = e_1$ approaches 0. See Appendix \ref{CD} for more details on the construction of critical values. By design, we have the following result.

\begin{cor} \label{cor2}
Given $\Gamma$ as defined in Appendix \ref{ACverification}, AsySz$_{LR\text{-}LF} = \alpha$ for $\alpha \in (0,1/2)$. 
\end{cor}
 
 \section{Monte Carlo} \label{MC}
 
In this section, we use simulations to assess (i) how well the finite-sample null rejection frequencies of the different tests, or testing procedures, that we consider are approximated by the foregoing asymptotic theory and (ii) how these tests compare in terms of (finite-sample) power. We generate data from the model given in \eqref{y}, \eqref{h}, and \eqref{phi}. We use a burn-in phase of 100 observations and set the starting values for the $y_t$ and $x_t$ series equal to zero. The sample size $n$ (after discarding the burn-in observations) is equal to 500. The number of simulations is 1,000. Table \ref{RF_table} shows the finite-sample rejection frequencies (in \%) of the likelihood ratio test for testing $H_0^\dagger$ ($LR^\dagger$), the two-step procedure ($TS$), the second testing procedure ($S$), and the test that uses $LR_n$ together with PI-LF ($LR$-$LF$) at the 5\% nominal level for different true parameter constellations. Throughout, $\zeta$ is set equal to 1, while $\pi$, $\kappa$, $\varphi$, $\beta_1$, and $\beta_2$ are varied as indicated in the table. 

\begin{table}[h!]											
\begin{center}									
\caption{Finite-sample rejection frequencies (in \%) at 5\% nominal level}			
\label{RF_table}									
\begin{tabular}{c|rrr|rrr|rrr|r|r}								
\hline								
\hline	
 & (1) & (2) & (3) & (4) & (5) & (6) & (7) & (8) & (9) & (10) & (11)  \\ 
\hline	
$\pi$ 	& \multicolumn{3}{c|}{0.2} & \multicolumn{3}{c|}{0.2} &\multicolumn{3}{c|}{0.2} & 0.64 & 0  \\ 
$\varphi$	& \multicolumn{3}{c|}{0.5} & \multicolumn{3}{c|}{0.5} &\multicolumn{3}{c|}{0.5}  & 0.5 & 0  \\ 
$\kappa$	& \multicolumn{3}{c|}{0} & \multicolumn{3}{c|}{0} &\multicolumn{3}{c|}{0}  & 0 & 0.99  \\ 
$\beta_1$ & \multicolumn{3}{c|}{0} & \multicolumn{3}{c|}{0.05} &\multicolumn{3}{c|}{0.1} & 0.11 & 0.3 \\ 
$\beta_2$ & 0 & 0.05 & 0.1 & 0 & 0.05 & 0.1 & 0 &0.05 & 0.1 & 0 & 0  \\ 
\hline
$LR^\dagger$ 	&  3.1  &    22.9  &  57.5 &  14.0  &  36.0 & 65.2  &  40.6 & 57.5 & 79.1 & 70.5 & 99.0 \\
$TS$		& 2.0   &    21.9 &   56.6 &   3.4 &   25.8 & 59.7   & 5.6 & 29.9 & 62.9 & 6.2 & 8.2\\
$S$ 		& 11.0   &  45.6 & 77.9  & \ 9.0    &  41.1 & 73.3  & 8.3 & 37.2 & 68.9 & 6.9 & 8.2	\\
$LR$-$LF$ 	& 5.3   & 34.5 &  67.7 &   4.8   &  30.9 & 61.9 & 4.0 & 26.8 & 57.1 & 3.9 & 0.0 \\
\hline								
\end{tabular}
\end{center}						
\end{table}

Columns (1), (4), (7), (10), and (11) of Table \ref{RF_table} report null rejection frequencies for $TS$, $S$, and $LR$-$LF$ ($\beta_2 = 0$). We see that the asymptotic null rejection probabilities that underlie the lower bounds on asymptotic size obtained in Section \ref{AsySz} provide good approximations to the corresponding finite-sample rejection frequencies. In particular, column (1) shows an 11\% rejection frequency of $S$, which is close to the corresponding asymptotic rejection probability of 9.48\%. Similarly, columns (10) and (11) show rejection frequencies of 6.2\% and 8.2\% for $TS$, which are close to the corresponding asymptotic rejection probabilities of 6\% and 6.65\%, respectively; note that $\sqrt{500} \times 0.11 \approx 2.5$ and $\sqrt{500} \times 0.3 \approx 6.7$.\footnote{Here, 6.7 appears to be large enough for the asymptotic theory obtained for $b_1 = \infty$ to provide a good approximation, which is corroborated by a rejection frequency of $LR^\dagger$ close to 100\%.} Furthermore, Table \ref{RF_table} reveals that $TS$ has null rejection frequencies below the nominal level for ``very small'' values of $\beta_1$ (see columns (1) and (4)); this is not surprising given the nature of the two-step procedure and the low rejection frequency of $LR^\dagger$ for such values of $\beta_1$. For the \textit{dgp}s considered in columns (1)--(9) our numerical evaluations show that $b_1 = 0$ is the (unique) least favorable configuration for $LR_n$. As a result, the $LR$-$LF$ offers sizeable power gains over $TS$ for ``very small'' values of $\beta_1$ (see columns (2), (3), (5), and (6)), where the latter, by continuity of the power curve, in some sense ``sacrifices'' power. Not surprisingly, for large(r) values of $\beta_1$ the power ranking is reversed (see columns (8) and (9)).

\bibliography{references}

\appendix

\section{Definition of $\Gamma$, verification of claims in Sections \ref{AT} and \ref{AsySz}, and proof of Proposition \ref{pro1}} \label{ACverification}

First, we define $\Gamma$. [Details to be added.]

The claim in the last paragraph before Section \ref{Results_close_to_zero} follows from Lemma 3.1 in AC by verifying Assumptions A and B3 in AC. [Details to be added.]

Next, we show that equation \eqref{quadratic_expansion} holds, which amounts to verifying Assumptions C1 and C4 in AC. [Details to be added.]

Next, we verify the weak convergence result for $G_n(\cdot)$, which amounts to verifying Assumptions C2 and C3 in AC. [Details to be added.] Note that
\[
	\Omega(\pi_1,\pi_2; \gamma_0) = E_{\gamma_0} l^\infty_{\psi,t}(\psi_0,\pi_1) l^\infty_{\psi,t}(\psi_0,\pi_2)' = \frac{c_0}{2\zeta_0^2} E_{\gamma_0} \frac{\partial h^\infty_t(\psi_0,\pi_1)}{\partial \psi} \frac{\partial h^\infty_t(\psi_0,\pi_2)}{\partial \psi'}.
\]
Then, equation \eqref{omega} is obtained by noting (i) 
\begin{align*}
	\frac{1}{\zeta_0^2} E_{\gamma_0} \sum_{i=0}^\infty \pi_1^i y^2_{t-i-1}  \sum_{j=0}^\infty \pi_2^j y^2_{t-j-1} & = 
	E_{\gamma_0} \sum_{i=0}^\infty \pi_1^i z^2_{t-i-1}  \sum_{j=0}^\infty \pi_2^j z^2_{t-j-1}\\ &= \sum_{i=0}^\infty \pi_1^i \pi_2^i E_{\gamma_0}  z^4_{t-i-1}  +  \sum_{i=0}^\infty \sum_{j\neq i} \pi_1^i \pi_2^j E_{\gamma_0} z^2_{t-i-1}  z^2_{t-j-1} \\ &= \sum_{i=0}^\infty \pi_1^i \pi_2^i E_{\gamma_0}  z^4_{t-i-1}  - \sum_{i=0}^\infty \pi_1^i \pi_2^i +  \sum_{i=0}^\infty \pi_1^i  \sum_{j=0}^\infty \pi_2^j \\ &= \frac{E_{\gamma_0}z_t^4-1}{1-\pi_1\pi_2} + \frac{1}{(1-\pi_1)(1-\pi_2)}\\
	&=  \frac{2c}{1-\pi_1\pi_2} + \frac{1}{(1-\pi_1)(1-\pi_2)},
\end{align*}
where the second to last equality uses $E_{\gamma_0} z^2_{t-i-1}  z^2_{t-j-1} = 1$ for $i \neq j$, (ii)
\[
\frac{1}{\zeta_0} E_{\gamma_0} \sum_{i=0}^\infty \pi_1^i y^2_{t-i-1} = E_{\gamma_0} \sum_{i=0}^\infty \pi_1^i z^2_{t-i-1} =   \sum_{i=0}^\infty \pi_1^i E_{\gamma_0} z^2_{t-i-1} = \frac{1}{1-\pi_1},
\]
and (iii)
\[
	\frac{1}{\zeta_0} E_{\gamma_0} \sum_{i=0}^\infty \pi_1^i y^2_{t-i-1}  \sum_{j=0}^\infty \pi_2^j x^2_{t-j-1} =  E_{\gamma_0} \sum_{i=0}^\infty \pi_1^i z^2_{t-i-1}  \sum_{j=0}^\infty \pi_2^j x^2_{t-j-1}.
\]

To show that $H(\pi;\gamma_0) = \Omega(\pi,\pi; \gamma_0)/c_0$, note that, by definition, we have
\[
	H(\pi;\gamma_0) = E_{\gamma_0} \frac{\partial}{\partial \psi'} \rho_{\psi,t}(\psi_0,\pi),
\]
where
\[
	\frac{\partial}{\partial \psi'} \rho_{\psi,t}(\theta) = \frac{\partial}{\partial \psi'} \left[ \frac{1}{2h^\infty_t(\theta)} \left( 1-  \frac{y_t^2}{h^\infty_t(\theta)} \right)  \frac{\partial h^\infty_t(\theta)}{\partial \psi} \right] =  \frac{\partial h^\infty_t(\theta)}{\partial \psi} \frac{\partial}{\partial \psi'} \left[ \frac{1}{2h^\infty_t(\theta)} \left( 1-  \frac{y_t^2}{h^\infty_t(\theta)} \right) \right],
\]
since $\frac{\partial h^\infty_t(\theta)}{\partial \psi}$ is not a function of $\psi$. Furthermore, we have
\[
	\frac{\partial}{\partial \psi'} \left[ \frac{1}{2h^\infty_t(\theta)} \left( 1-  \frac{y_t^2}{h^\infty_t(\theta)} \right) \right]
	= \left[ - \frac{1}{2(h^\infty_t(\theta))^2} \left( 1-  \frac{y_t^2}{h^\infty_t(\theta)} \right) + \frac{1}{2h^\infty_t(\theta)} \frac{y_t^2}{(h^\infty_t(\theta))^2} \right] \frac{\partial h^\infty_t(\theta)}{\partial \psi'}.
\]
Now, given $\Gamma$, we have
\[
	E_{\gamma_0} - \frac{1}{2\zeta_0^2} \left( 1-  z_t^2 \right) \frac{\partial h^\infty_t(\psi_0,\pi)}{\partial \psi}  \frac{\partial h^\infty_t(\psi_0,\pi)}{\partial \psi'} = 0
\]
and the desired result follows.

Next, we derive $K(\pi;\gamma_0) = K(\psi_0,\pi;\gamma_0)$. The expected value entering $K_n(\theta,\gamma^*)$, see equation (3.6) in AC, is given by 
\begin{align*}
 &E_{\gamma^*} \frac{1}{2h^\infty_t(\theta)} \left( 1-  \frac{y_t^2}{h^\infty_t(\theta)} \right)  \frac{\partial h^\infty_t(\theta)}{\partial \psi} = E_{\gamma^*} \frac{1}{2h^\infty_t(\theta)} \left( 1-  \frac{h^\infty_t(\theta^*)}{h^\infty_t(\theta)} \right)  \frac{\partial h^\infty_t(\theta)}{\partial \psi} \\
	=& E_{\gamma^*} \frac{1}{2h^\infty_t(\theta)} \frac{\partial h^\infty_t(\theta)}{\partial \psi} - E_{\gamma^*} \frac{1}{2h^\infty_t(\theta)}   \frac{h^\infty_t(\theta^*)}{h^\infty_t(\theta)}   \frac{\partial h^\infty_t(\theta)}{\partial \psi} \\
	=& E_{\gamma^*} \frac{1}{2h^\infty_t(\theta)} \frac{\partial h^\infty_t(\theta)}{\partial \psi} - E_{\gamma^*} \frac{1}{2h^\infty_t(\theta)}   \frac{\zeta^* +  \beta^*_1 \sum_{i=0}^{\infty} \pi^{*i} y_{t-i-1}^2 + \beta^*_2 \sum_{i=0}^{\infty} \pi^{*i} x_{t-i-1}^2}{h^\infty_t(\theta)}   \frac{\partial h^\infty_t(\theta)}{\partial \psi}
\end{align*}
such that
\[
	 K(\psi_0,\pi;\gamma_0) = - E_{\gamma_0} \frac{1}{2\zeta_0^2} \left. \frac{\partial h^\infty_t(\theta)}{\partial \psi} \right|_{\theta = (\psi_0,\pi)} \left[  \sum_{j=0}^{\infty} \pi^j_0 y_{t-j-1}^2 ; \sum_{j=0}^{\infty} \pi^j_0 x_{t-j-1}^2 \right], 
\]
cf.\ Assumption C5 in AC. The desired result then follows using similar observations to those used in the derivation of $\Omega(\pi_1,\pi_2; \gamma_0)$.

Equation \eqref{dist_Z} follows from Lemma 9.1 in AC, upon verification of Assumptions C2, C3, and C5 in AC. [Details to be added.]

Equation \eqref{lambda_hat} follows from arguments analogous to those underlying Theorem 1(a) in \cite{Andrews:01}. [Details to be added.] Similarly, equations \eqref{asy_dist_estimator} and \eqref{asy_dist_objective} follow from combining the arguments underlying Theorem 3.1 in AC and Theorem 1(b) and (c) in \cite{Andrews:01}. [Details to be added.] Furthermore, equations \eqref{asy_dist_LR_star} and \eqref{asy_dist_LR} follow from arguments analogous to those underlying Theorem 2(b) in \cite{Andrews:01}. [Details to be added.] 

Equation \eqref{Lemma32b} follows from Lemma 3.2(b) in AC. [Details to be added.]  To gain intuition for this result, note that, under $\{\gamma_n\} \in \Gamma(\gamma_0,\infty,b,\omega_0,p)$ with $\beta_0 = 0$, the equivalent of the above $q(\cdot)$ function is given by
\begin{equation} \label{equivalent_q}
	(\lambda - \{- H^{-1}(\pi;\gamma_0)  K(\pi;\gamma_0) \omega_0\} )' H(\pi;\gamma_0)(\lambda - \{- H^{-1}(\pi;\gamma_0)  K(\pi;\gamma_0) \omega_0\} ),
\end{equation}
which is minimized over $\Lambda$, since
\[
	\frac{1}{\| \beta_n \|} (\Psi - \psi_{0,n}) \to \Lambda.
\]
Then, noting that $- H^{-1}(\pi;\gamma_0)  K(\pi;\gamma_0) \omega_0 \geq 0$ (since all elements of $K(\pi;\gamma_0)$ are nonpositive), we have that the minimum of \eqref{equivalent_q} over $\Lambda$ is equal to zero. 

Next, we show that $\eta(\pi;\gamma_0,\omega_0)$ is uniquely minimized at $\pi = \pi_0 \ \forall \gamma_0 \in \Gamma$ with $\beta_0 = 0$. [Details to be added.]

The two claims in the last sentence of Section \ref{results_b_infty} follow from Lemmas 3.3 and 3.4 in AC, respectively. [Details to be added.]

Next, we show that equation \eqref{quadratic_expansion_full_vector} holds. [Details to be added.]

Equation \eqref{Z_infty} holds given the definition of $\Gamma$. [Details to be added.]

Equations \eqref{lambda_infty} and \eqref{obj_fun_infty} follow from combining the arguments underlying Theorem 3.2 in AC and Theorem 1(b) and (c) in \cite{Andrews:01}. [Details to be added.] Similarly, equation \eqref{LR_infty} follows from arguments analogous to those underlying Theorem 2(b) in \cite{Andrews:01}. [Details to be added.]

\begin{proof}[Proof of Proposition \ref{pro1}]
[Details to be added.]
\end{proof}

Next, we show that $LR^\dagger_n \to \infty$ and $P(\widetilde{LR}^\dagger_{n,1-\alpha} < \infty) = 1$ under $\{\gamma_n\} \in \Gamma(\gamma_0,\infty,b,\omega_0,p)$ with $\pi_0 = 0$. [Details to be added.]

Lastly, we show that $LR_{1-\alpha}^\infty(\gamma_0,p)$ is strictly monotonically increasing (decreasing) in $p$ if $ \rho(\gamma_0) < 0$ ($ \rho(\gamma_0) > 0$). [Details to be added.]

\section{Computational details} \label{CD}

In what follows, we let $\bar{\Pi}$ and $\bar{\Pi}^*$ denote discretized versions of $\Pi$ and $\Pi^*$, respectively. 

\subsection{Evaluation of asymptotic distribution for $\| b \| < \infty$} 

The asymptotic distributions displayed in Figures \ref{plot_pi_n_500}--\ref{plot_beta_zeta_n_500} as well as the asymptotic rejection probabilities $RP_{TS}(h)$ and $RP_{S}(h)$ are evaluated using simulation. The underlying \textit{dgp} is given in \eqref{y}, \eqref{h}, and \eqref{phi}. To simulate from $G(\cdot;\gamma_0)$, we use the method of \cite{Hansen:96}. In particular, we obtain the $j^\text{th}$ draw from $G(\pi;\gamma_0)$ by computing
\[
	\frac{1}{\sqrt{N}}\sum_{k=1}^N Z^j_k \left(\sum_{i=0}^{99} \pi^i z^2_{k,t-i-1},\sum_{i=0}^{99} \pi^i x^2_{k,t-i-1}/\zeta_0,1/\zeta_0\right)'{\sqrt{c_0/2}},
\]
 where $\{ Z^j_k \}_{k=1}^N$ are iid $N(0,1)$, independent across $j$, and where $\{ z_{k,t-i-1}, x_{k,t-i-1} \}_{i=0}^{99}$ are generated according to \eqref{phi} (with $z_{k,t-i-1} = z_{t-i-1}$ and $x_{k,t-i-1} = x_{t-i-1}$) using a burn-in phase of 100 observations. Here and in what follows, $c_0 = 1$.
 
We approximate $H(\pi;\gamma_0)$ and $K(\pi;\gamma_0)$ by

{\scriptsize
\[\frac{1}{2} \left[ \begin{array}{ccc}   \frac{2{c_0}}{1-\pi} + \frac{1}{(1-\pi)^2}  & \frac{1}{\zeta_0}\frac{1}{N^*} \sum_{l=1}^{N^*} \sum_{i=0}^{99} \pi^i z^2_{l,t-i-1}  \sum_{j=0}^{99} \pi^j x^2_{l,t-j-1} &  \frac{1}{\zeta_0} \frac{1}{1-\pi} \\  \frac{1}{\zeta_0}\frac{1}{N^*} \sum_{l=1}^{N^*} \sum_{i=0}^{99} \pi^i z^2_{l,t-i-1}  \sum_{j=0}^{99} \pi^j x^2_{l,t-j-1} &  \frac{1}{\zeta_0^2}\frac{1}{N^*} \sum_{l=1}^{N^*}   ( \sum_{i=0}^{99} \pi^i x^2_{l,t-i-1} )^2 &  \frac{1}{\zeta_0^2}\frac{1}{N^*} \sum_{l=1}^{N^*}    \sum_{i=0}^{99} \pi^i x^2_{l,t-i-1}  \\  \frac{1}{\zeta_0} \frac{1}{1-\pi} &\frac{1}{\zeta_0^2}\frac{1}{N^*} \sum_{l=1}^{N^*}    \sum_{i=0}^{99} \pi^i x^2_{l,t-i-1} & \frac{1}{\zeta_0^2} \end{array} \right]
\]}and
{\scriptsize
\[
	 - \frac{1}{2} \left[ \begin{array}{cc}   \frac{2{c_0}}{1-\pi\pi_0} + \frac{1}{(1-\pi)(1-\pi_0)}  & \frac{1}{\zeta_0}\frac{1}{N^*} \sum_{l=1}^{N^*} \sum_{i=0}^{99} \pi^i z^2_{l,t-i-1}  \sum_{j=0}^{99} \pi_0^j x^2_{l,t-j-1} \\  \frac{1}{\zeta_0}\frac{1}{N^*} \sum_{l=1}^{N^*} \sum_{i=0}^{99} \pi_0^i z^2_{l,t-i-1}  \sum_{j=0}^{99} \pi^j x^2_{l,t-j-1} &  \frac{1}{\zeta_0^2}\frac{1}{N^*} \sum_{l=1}^{N^*} \sum_{i=0}^{99} \pi^i x^2_{l,t-i-1}  \sum_{j=0}^{99} \pi_0^j x^2_{l,t-j-1} \\  \frac{1}{\zeta_0} \frac{1}{1-\pi_0} &\frac{1}{\zeta_0^2}\frac{1}{N^*} \sum_{l=1}^{N^*}    \sum_{i=0}^{99} \pi_0^i x^2_{l,t-i-1} \end{array} \right],
\]}respectively, where $\{ z_{l,t-i-1}, x_{l,t-i-1} \}_{i=0}^{99}$ are generated similarly to $\{ z_{k,t-i-1}, x_{k,t-i-1} \}_{i=0}^{99}$.

Together, this allows us to obtain the $j^\text{th}$ draw from $Z(\pi;\gamma_0,b)$. The corresponding draws of the ``asymptotic versions'' of the relevant estimators and test statistics can then be obtained using the formulae in Section \ref{ad_b_less_infty} with $\bar{\Pi}$ in place of $\Pi$ for a total of $J$ draws. 

We use $J = N =$ 10,000, $N^* =$ 100,000, and $\bar{\Pi} = \{0,0.01,0.02,\dots,\overline{\pi}\}$ with $\overline{\pi} = 0.9$.

\subsection{Evaluation of $\rho(\gamma_0)$}
In order to evaluate $\rho(\gamma_0)$, it suffices to evaluate $J(\gamma_0,e_1)$ with $\beta_{0,2} = \pi_0 = 0$. As before, the underlying \textit{dgp} is given in \eqref{y}, \eqref{h}, and \eqref{phi}. We approximate $J(\gamma_0,e_1)$ with $\beta_{0,2} = \pi_0 = 0$ by
\[
	\frac{1}{2N^*} \sum_{l=1}^{N^*}  \frac{\tau_l(e_1,0)}{h^\infty_{l,t}(\theta_0)}  \frac{\tau'_l(e_1,0)}{h^\infty_{l,t}(\theta_0)},
\]
where $\tau_l(e_1,0) = \left( y_{l,t-1}^2, x_{l,t-1}^2, 1,  y_{l,t-2}^2 \right)'$ and $h_{l,t}^\infty(\theta_0) = \zeta_0 +  \beta_{0,1}  y_{l,t-1}^2$. Here, $\{ y_{l,t-i-1}, x_{l,t-i-1} \}_{i=0}^{1}$ is generated according to \eqref{y}, \eqref{h}, and \eqref{phi} (with $y_{l,t-i-1} = y_{t-i-1}$ and $x_{l,t-i-1} = x_{t-i-1}$) using a burn-in phase of 100 observations.

We use $N^* =$ 100,000.

\subsection{Computation of critical values} 

\subsubsection{$b_1 < \infty$ and critical value for testing $H_0^\dagger$ using $LR_n^\dagger$}

Under $\{\gamma_n\} \in \Gamma(\gamma_0,0,b)$ with $b_1 < \infty$ and $b_2 = 0$, $\theta_0$ is consistently estimated by $\hat{\theta}^\dagger_{n,0}$ (as $\beta_0 = 0_2$). Let $\tilde{G}(\cdot)$ denote $G(\cdot;\gamma_0)$ with unknown quantities replaced by consistent estimators, i.e., by estimators that are consistent under $\{\gamma_n\} \in \Gamma(\gamma_0,0,b)$ with $ b_1 < \infty$ and $b_2 = 0$. 
To simulate from $\tilde{G}(\cdot)$, we use the method of \cite{Hansen:96}. In particular, we obtain the $j^\text{th}$ draw from $\tilde{G}(\pi)$ by computing
\[
	\tilde{G}_j(\pi) = \frac{1}{\sqrt{n}}\sum_{t=1}^n Z^j_t \left(\sum_{i=0}^{t-1} \pi^i y^2_{t-i-1},\sum_{i=0}^{t-1} \pi^i x^2_{t-i-1},1\right)'{\sqrt{\hat{c}^\dagger_n/2}}/\hat{\zeta}^\dagger_{n,0},
\]
 where $\{ Z^j_t \}_{t=1}^n$ are iid $N(0,1)$, independent across $j$. Here, we use the fact that $h_t(\hat{\theta}^\dagger_{n,0}) = \hat{\zeta}^\dagger_{n,0}$. Furthermore, we estimate $H(\pi;\gamma_0)$ and $K(\pi;\gamma_0)$ by

{\scriptsize
\begin{align*}
	 &\tilde{H}(\pi) = \\ &\frac{1}{2} \left[ \begin{array}{ccc}   \frac{2{\hat{c}^\dagger_n}}{1-\pi} + \frac{1}{(1-\pi)^2}  & \frac{1}{(\hat{\zeta}^\dagger_{n,0})^2}\frac{1}{n} \sum_{t=1}^n \sum_{i=0}^{t-1} \pi^i y^2_{t-i-1}  \sum_{j=0}^{t-1} \pi^j x^2_{t-j-1} &  \frac{1}{\hat{\zeta}^\dagger_{n,0}} \frac{1}{1-\pi} \\  \frac{1}{(\hat{\zeta}^\dagger_{n,0})^2}\frac{1}{n} \sum_{t=1}^n \sum_{i=0}^{t-1} \pi^i y^2_{t-i-1}  \sum_{j=0}^{t-1} \pi^j x^2_{t-j-1} &  \frac{1}{(\hat{\zeta}^\dagger_{n,0})^2}\frac{1}{n} \sum_{t=1}^n   ( \sum_{i=0}^{t-1} \pi^i x^2_{t-i-1} )^2 &  \frac{1}{(\hat{\zeta}^\dagger_{n,0})^2}\frac{1}{n} \sum_{t=1}^n    \sum_{i=0}^{t-1} \pi^i x^2_{t-i-1}  \\  \frac{1}{\hat{\zeta}^\dagger_{n,0}} \frac{1}{1-\pi} &\frac{1}{(\hat{\zeta}^\dagger_{n,0})^2}\frac{1}{n} \sum_{t=1}^n    \sum_{i=0}^{t-1} \pi^i x^2_{t-i-1} & \frac{1}{(\hat{\zeta}^\dagger_{n,0})^2} \end{array} \right]
\end{align*}}and
{\scriptsize
\begin{align*}
	& \tilde{K}(\pi,\pi_0) = \\  & - \frac{1}{2} \left[ \begin{array}{cc}   \frac{2{\hat{c}^\dagger_n}}{1-\pi\pi_0} + \frac{1}{(1-\pi)(1-\pi_0)}  & \frac{1}{(\hat{\zeta}^\dagger_{n,0})^2}\frac{1}{n} \sum_{t=1}^n \sum_{i=0}^{t-1} \pi^i y^2_{t-i-1}  \sum_{j=0}^{t-1} \pi_0^j x^2_{t-j-1} \\   \frac{1}{(\hat{\zeta}^\dagger_{n,0})^2}\frac{1}{n} \sum_{t=1}^n \sum_{i=0}^{t-1} \pi_0^i y^2_{t-i-1}  \sum_{j=0}^{t-1} \pi^j x^2_{t-j-1} &  \frac{1}{(\hat{\zeta}^\dagger_{n,0})^2}\frac{1}{n} \sum_{t=1}^n \sum_{i=0}^{t-1} \pi^i x^2_{t-i-1}  \sum_{j=0}^{t-1} \pi_0^j x^2_{t-j-1} \\  \frac{1}{\hat{\zeta}^\dagger_{n,0}} \frac{1}{1-\pi_0} & \frac{1}{(\hat{\zeta}^\dagger_{n,0})^2}\frac{1}{n} \sum_{t=1}^n    \sum_{i=0}^{t-1} \pi_0^i x^2_{t-i-1} \end{array} \right],
\end{align*}}respectively. Then, the $j^\text{th}$ draw from $\tilde{Z}(\pi;\pi_0,(b_1,0)')$, the ``estimated counterpart'' to $Z(\pi;\gamma_0,(b_1,0)')$, is given by
\[
	\tilde{Z}_j(\pi;\pi_0,(b_1,0)') = - \tilde{H}^{-1}(\pi) \{ \tilde{G}_j(\pi) + \tilde{K}(\pi,\pi_0)(b_1,0)' \}.
\]
The $j^\text{th}$ draw of $\widetilde{LR}^\dagger$, the ``estimated counterpart'' to $LR^\dagger(\gamma_0,0_2)$, can then be approximated by evaluating $\tilde{Z}_j(\pi;\pi_0,0_2)$ over $\bar{\Pi}$ using the corresponding formulae in Section \ref{ad_b_less_infty}. We note that $\tilde{Z}_j(\pi;\pi_0,0_2)$ and thus the $j^\text{th}$ draw of $\widetilde{LR}^\dagger$ do not depend on $\pi_0$. Finally, the $1-\alpha$ quantile of the empirical distribution of $J$ draws from $\widetilde{LR}^\dagger$ provides us with an estimate of the $1-\alpha$ quantile of the asymptotic distribution of $LR^\dagger_n$ under $\{\gamma_n\} \in \Gamma(\gamma_0,0,b)$ with $b = 0_2$. We denote this quantile by $\widetilde{LR}^\dagger_{n,1-\alpha}$.

The $j^\text{th}$ draw of $\widetilde{LR}(\pi_0,(b_1,0)')$, the ``estimated counterpart'' to $LR(\gamma_0,(b_1,0)')$, is approximated in a similar fashion. And we take the $1-\alpha$ quantile of the empirical distribution of $J$ draws from $\widetilde{LR}(\pi_0,(b_1,0)')$ as estimator of the $1-\alpha$ quantile of the asymptotic distribution of $LR_n$ under $\{\gamma_n\} \in \Gamma(\gamma_0,0,b)$ with $b_1 < \infty$ and $b_2 = 0$, which we denote by $\widetilde{LR}_{n,1-\alpha}(\pi_0,b_1)$.

\subsubsection{$b_1 = \infty$}
Under $\{\gamma_n\} \in \Gamma(\gamma_0,\infty,b,\omega_0,p)$ with $b_1 = \infty$, $b_2 = 0$, and $\omega_0 = e_1$, $\theta_0$ is consistently estimated by $\hat{\theta}_{n,0}$. Here, we do not impose $\pi_0 = 0$ and estimate $J(\gamma_0,e_1)$ by
\[
	\tilde{J} = \frac{1}{2n} \sum_{t=1}^n \tilde{\tau}_t \tilde{\tau}'_t/ h^2_t(\hat{\theta}_{n,0}),
\]
where
\[
	\tilde{\tau}_t = \left( \sum_{i=0}^{t-1} \hat{\pi}_{n,0}^i y_{t-i-1}^2, \sum_{i=0}^{t-1} \hat{\pi}_{n,0}^i x_{t-i-1}^2, 1, \sum_{i=1}^{t-1} i \hat{\pi}_{n,0}^{i-1} y_{t-i-1}^2   \right)'.
\]
Then, we estimate $\rho(\gamma_0)$ by
\[
	\tilde \rho = \frac{\left( \tilde J^{-1}\right)_{\beta_2,\pi}}{\sqrt{\left( \tilde J^{-1}\right)_{\beta_2,\beta_2}\left( \tilde J^{-1}\right)_{\pi,\pi}}}.
\]
Finally, let $\tilde q = q(\tilde \rho)$ and define
\[
	\widetilde{LR}^\infty_{n,1-\alpha} = \begin{cases} \text{cv}_{1-\alpha} & \text{if } \tilde \rho < 0 \\ \left( \sum_{i=0}^2 \mathbbm{1} (W = i) \times \chi^2_i\right)_{1-\alpha} &  \text{if } \tilde \rho \geq 0 \end{cases},
\]
where $\{ \chi^2_i \}_{i=0}^2$ are independent $\chi^2$ random variables with degree of freedom $i=0,1,$ and 2 that are independent of the random variable $W$ that is equal to 0, 1, and 2 with probabilities $1/2 - \tilde q$, 1/2, and $\tilde q$, respectively. Here, the subscript ${1-\alpha}$ indicates the $1-\alpha$ quantile of the random variable in parentheses. Note that under $\{\gamma_n\} \in \Gamma(\gamma_0,\infty,b,\omega_0,p)$ with $\pi_0 = 0$, $b_1 = \infty$, $b_2 = 0$, and $\omega_0 = e_1$ such that $\rho(\gamma_0) \geq 0$, $\widetilde{LR}^\infty_{n,1-\alpha}$ consistently estimates the $1-\alpha$ quantile of $LR^\infty(\gamma_0,e_1,(\infty,0)',0)$. 

\subsubsection{Plug-in least favorable configuration critical value}
The plug-in least favorable configuration critical value that we propose is given by
\[
	\text{cv}^\text{PI-LF}_{n,1-\alpha} = \max \{ \max_{(\pi_0,b_1) \in {PB}} \widetilde{LR}_{n,1-\alpha}(\pi_0,b_1), \mathbbm{1}(\hat{\beta}_{n,1} > 0) \cdot \widetilde{LR}^\infty_{n,1-\alpha}, \mathbbm{1}(\hat{\beta}_{n,1} > 0) \cdot \text{cv}_{1-\alpha} \},
\]
where
\[
	PB = \{ (\pi_0,b_1) \in \bar{\Pi}^* \times B_1: 3(\sqrt{n}b_1)^2+\pi_0^2+2(\sqrt{n}b_1)\pi_0 < 1 \}
\]
and where $B_1$ is a finite grid. 

In our simulations, we use $J =$ 10,000, $\bar{\Pi} = \{0,0.1,0.2,\dots,0.9\}$, $\bar{\Pi}^* = \{0,0.1,0.2,\dots,0.8\}$, and 
$$B_1 = \{ 0, 0.1, 0.2, 0.3, 0.4, 0.5, 1, 1.5, 2, 3, 4, 5, 6, 8, 10, 12 \}.$$
We chose the grids relatively coarsely to keep the computation cost of the Monte Carlo simulations low. For empirical applications, we recommend the use of finer grids. We also note that there is no need to include values of $b_1$ greater than $\sqrt{n/3} $ in $B_1$.

\newpage

\section{Additional graphs} \label{additional_graphs}
 
\begin{figure}[h!] 
  \begin{center}
    \includegraphics[width=54mm]{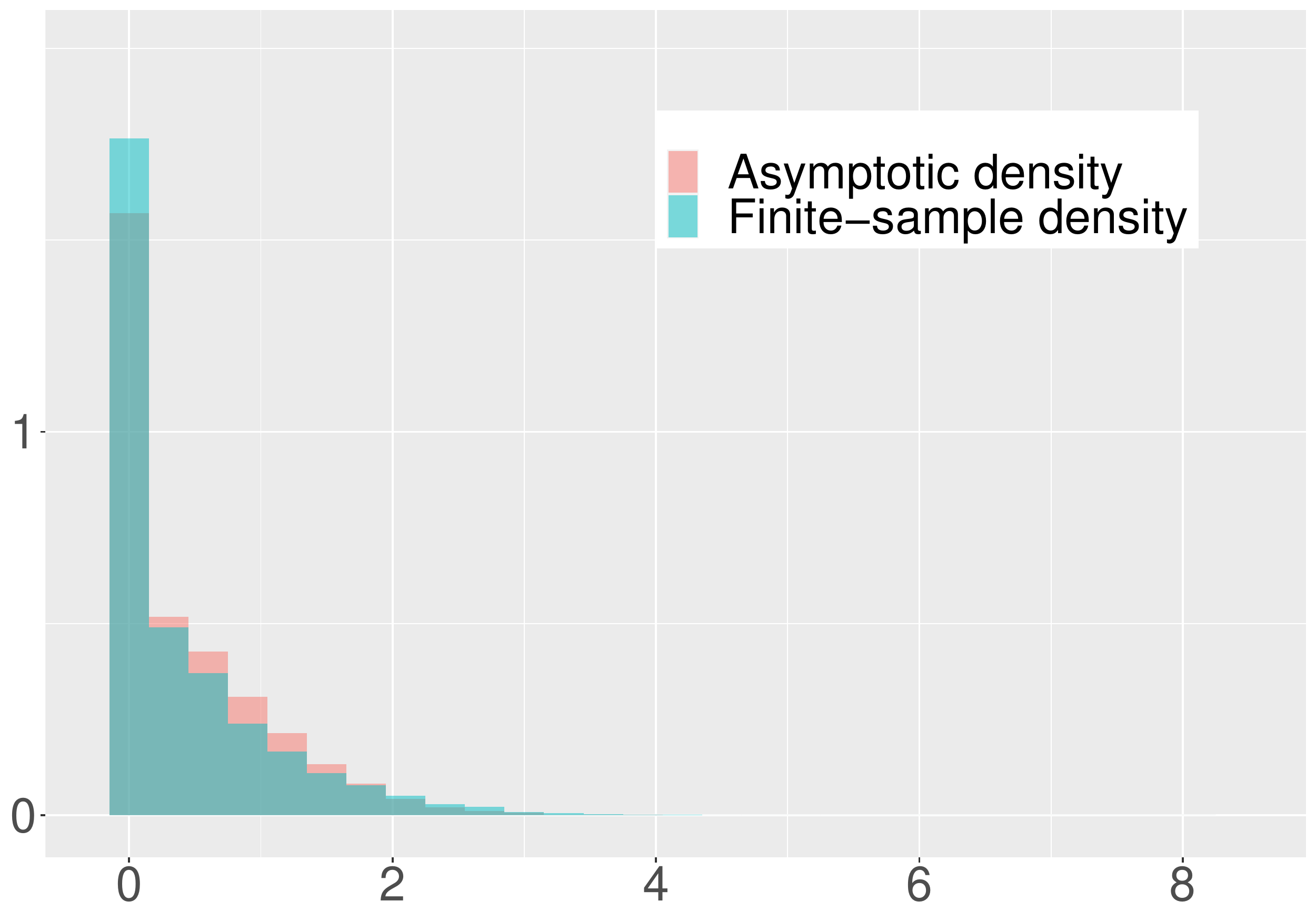}
    \includegraphics[width=54mm]{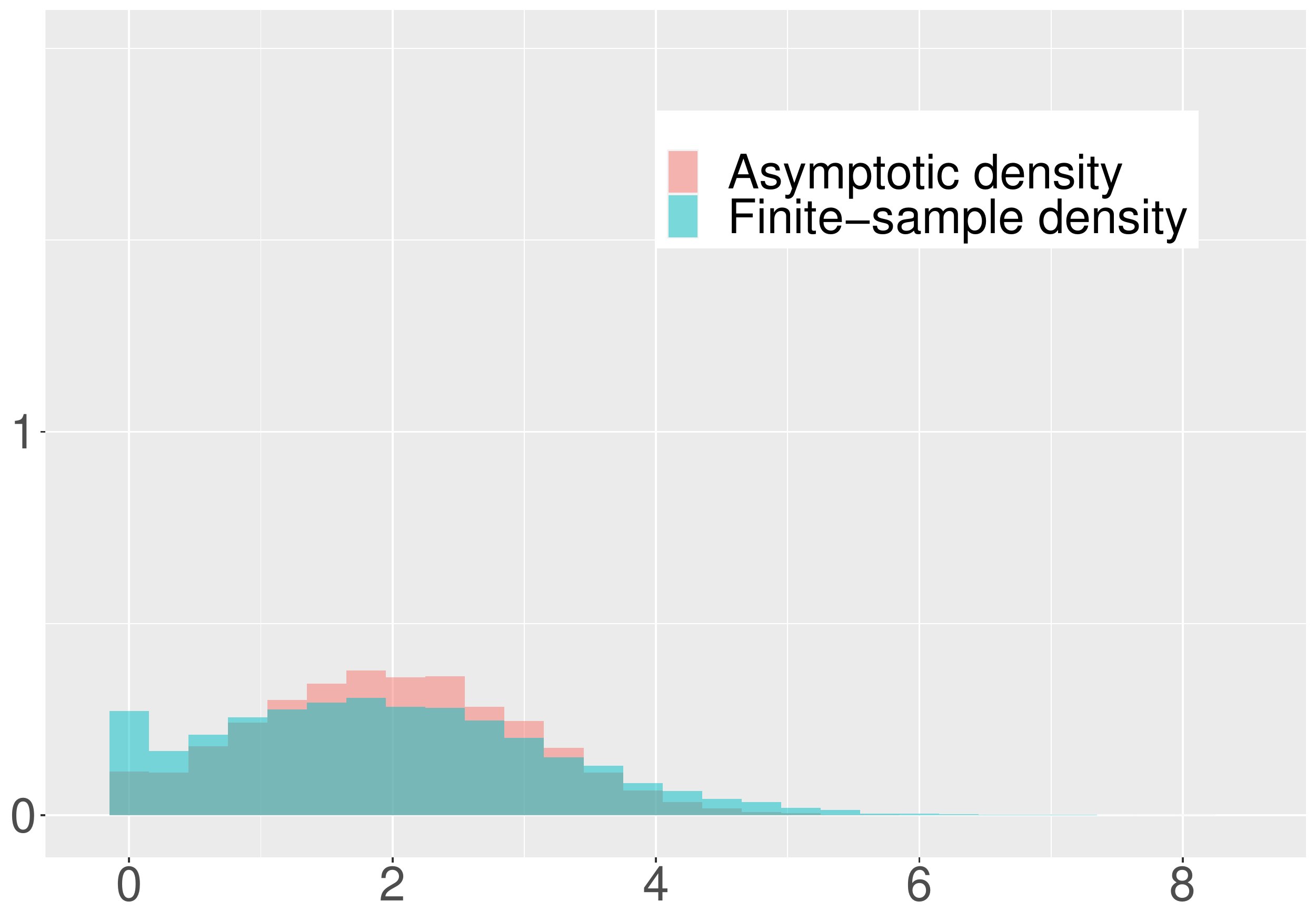}
    \includegraphics[width=54mm]{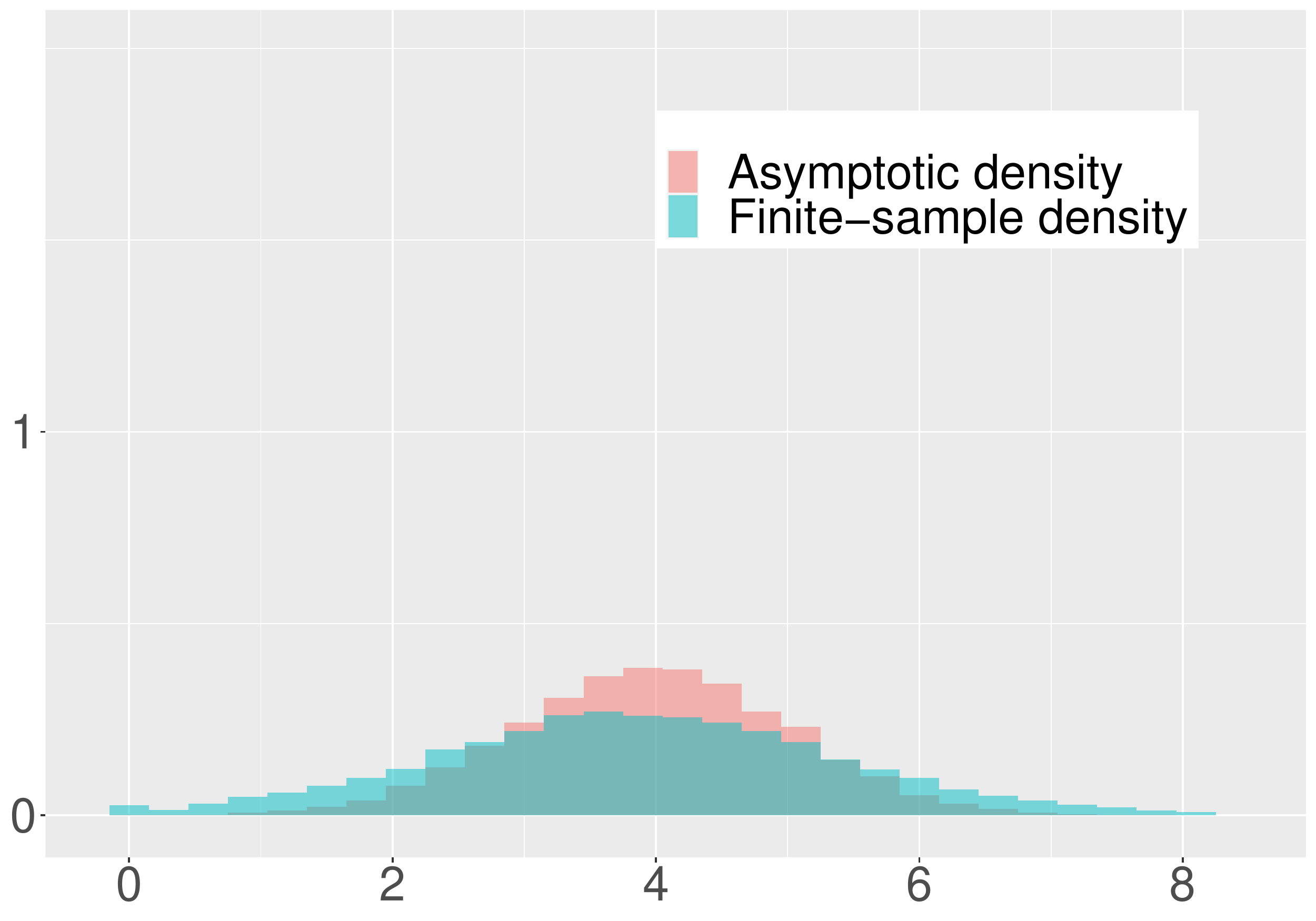}
        \includegraphics[width=54mm]{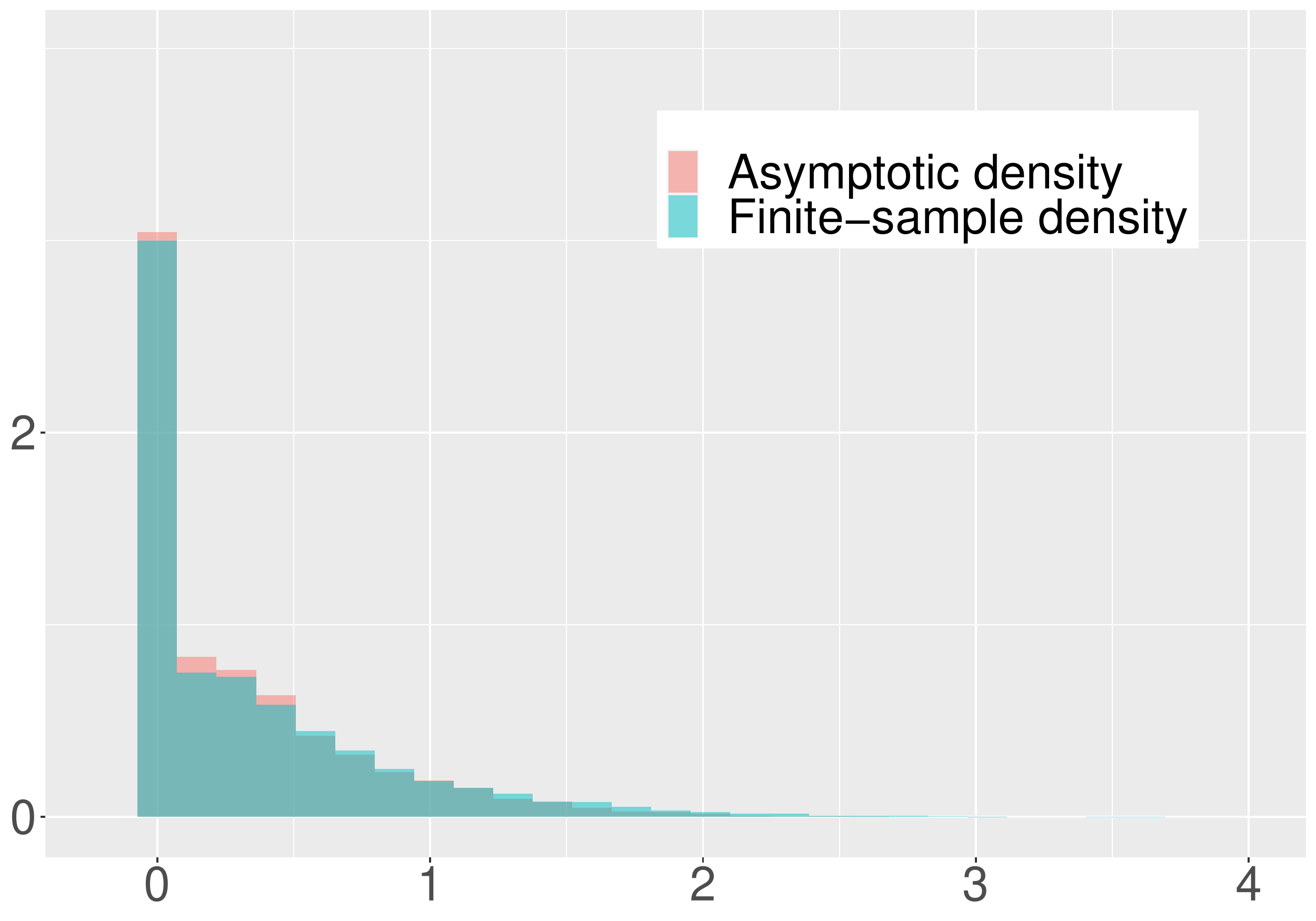}
    \includegraphics[width=54mm]{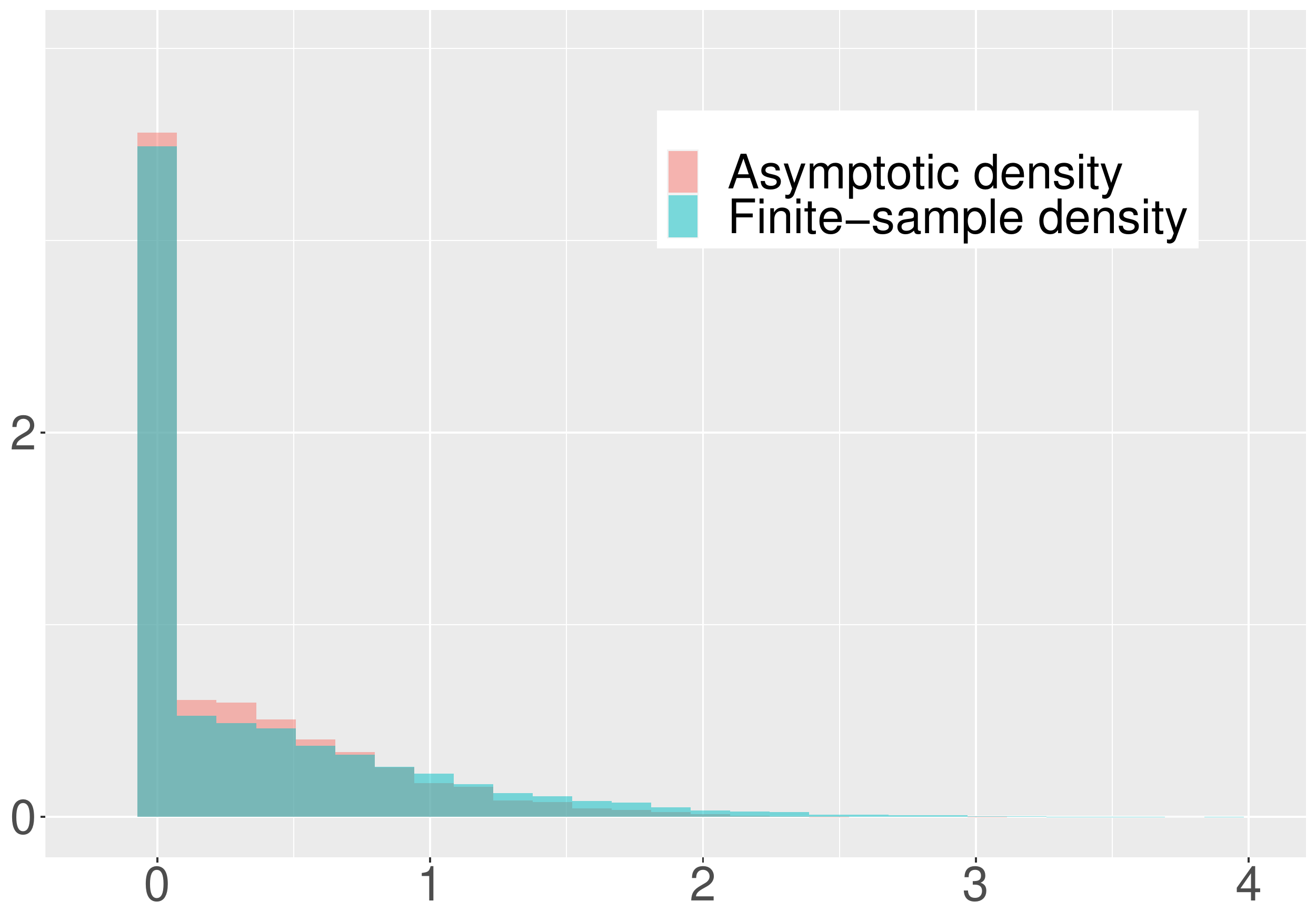}
    \includegraphics[width=54mm]{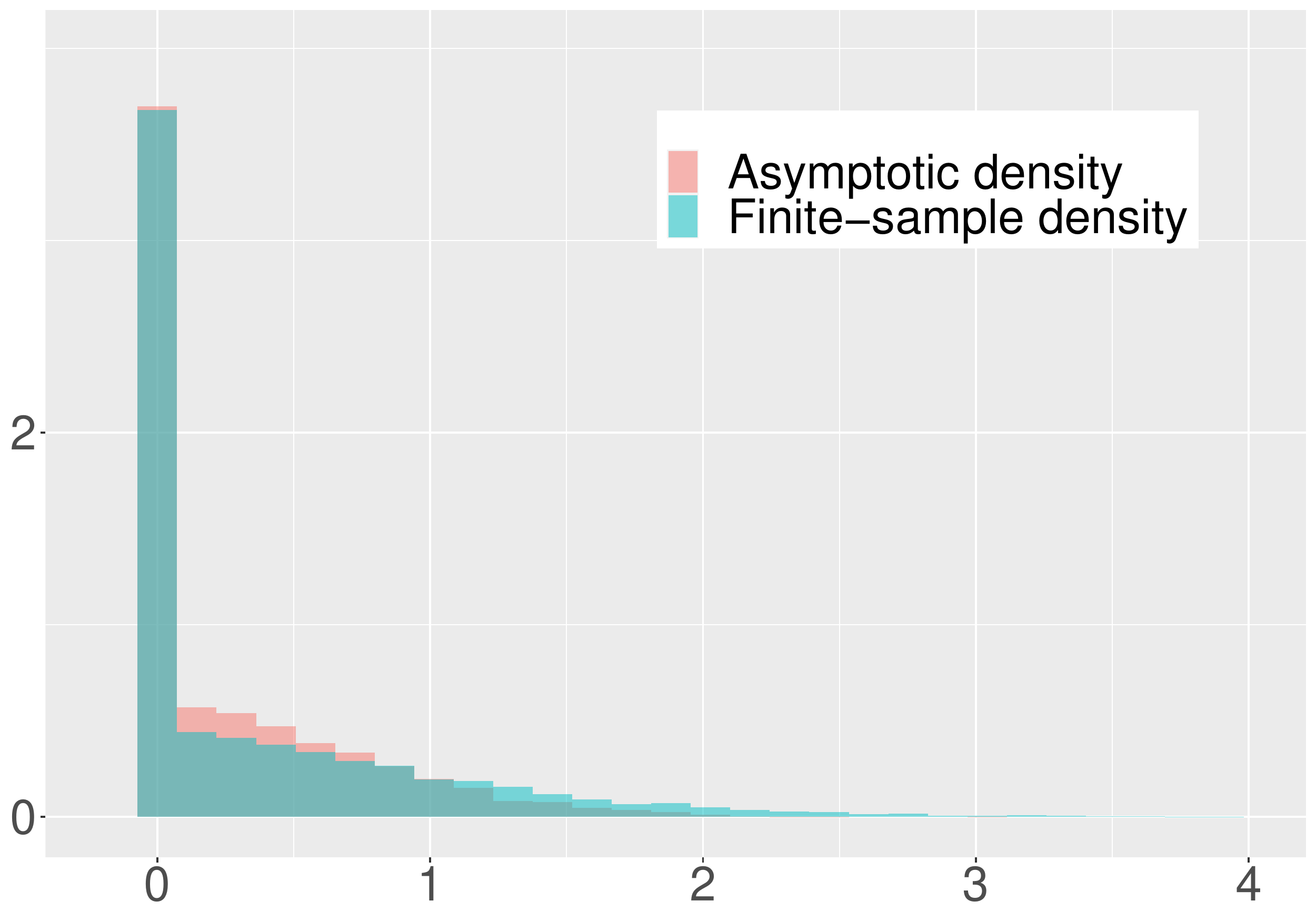}
        \includegraphics[width=54mm]{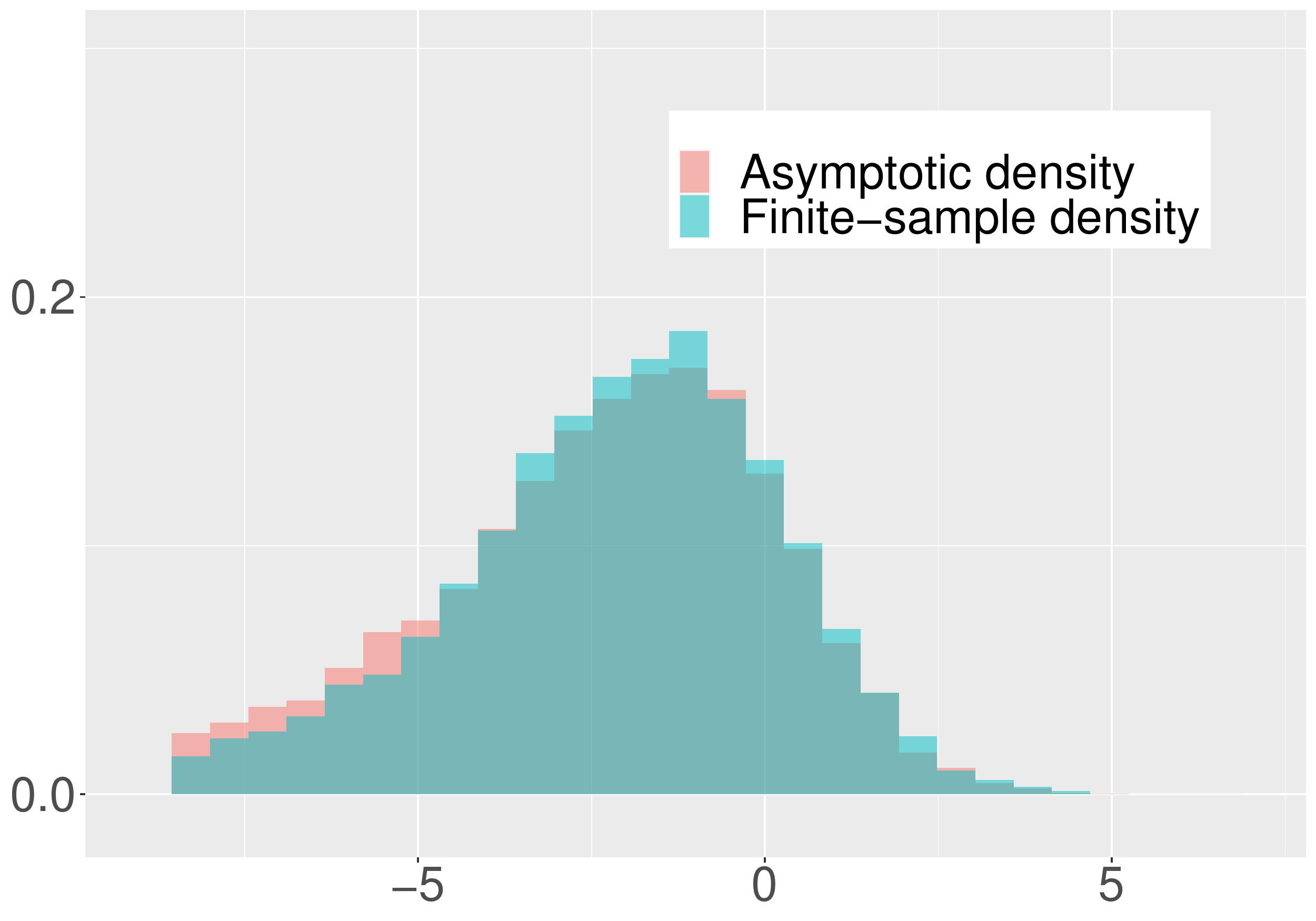}
    \includegraphics[width=54mm]{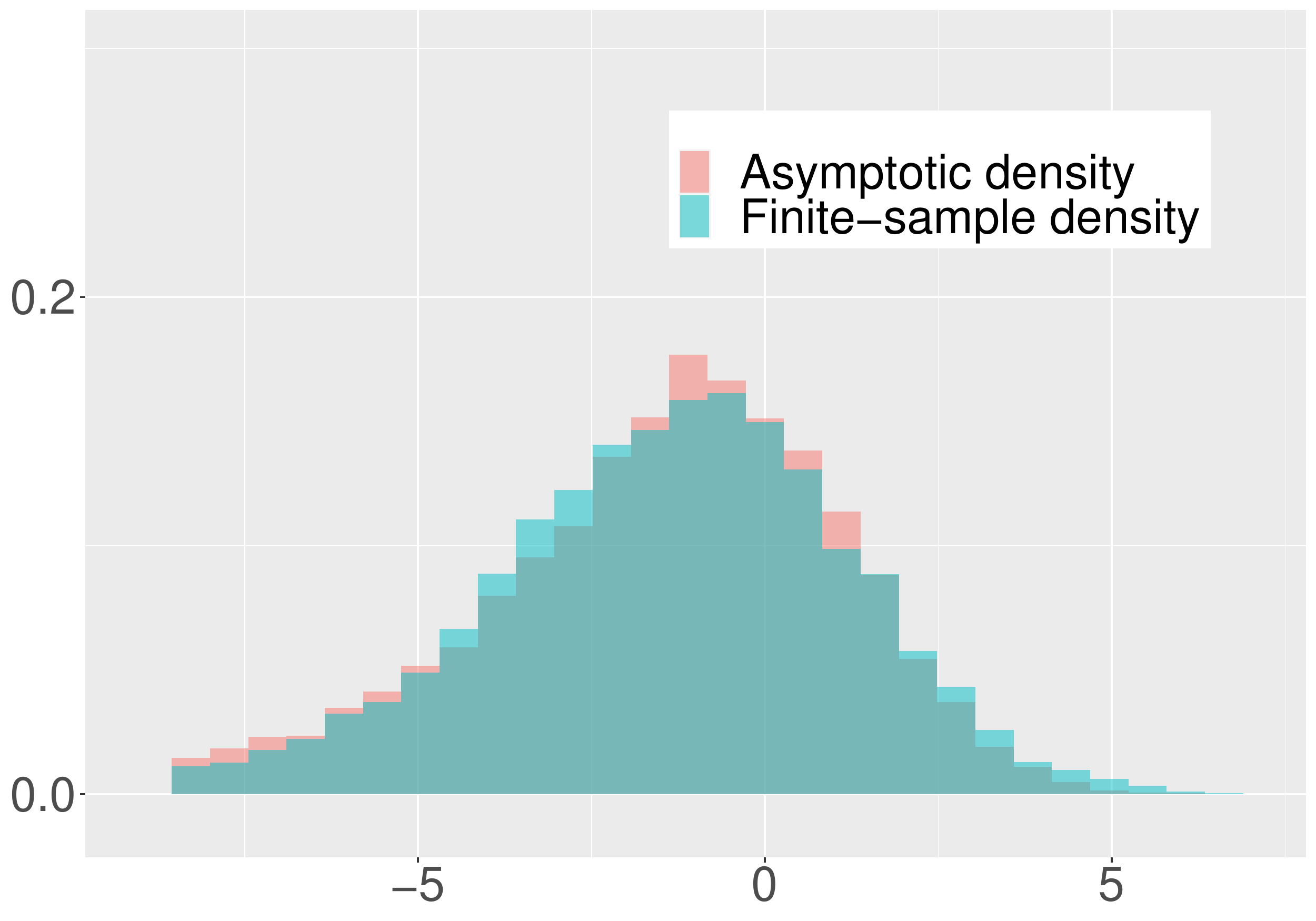}
    \includegraphics[width=54mm]{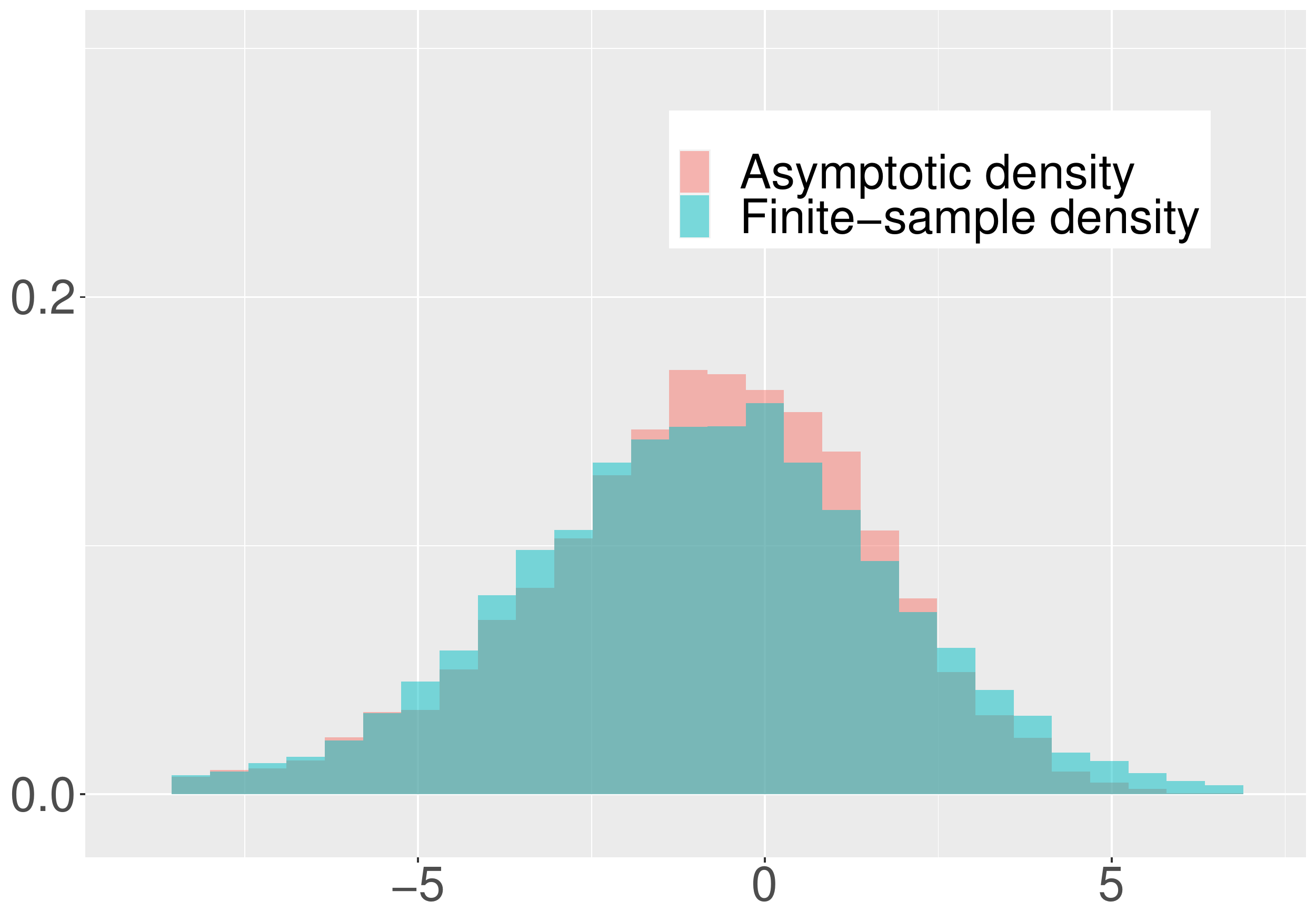}     
     \caption{Asymptotic and finite-sample ($n = 500$) densities of $\sqrt{n}\hat{\beta}_{n,1}$, $\sqrt{n}\hat{\beta}_{n,2}$, and $\sqrt{n}(\hat{\zeta}_n-\zeta_n)$ from top to bottom with $\sqrt{n} \beta_1 = b_1 = 0,2,$ and 4 from left to right and $\beta_2 = b_2 = 0$. Here, $\pi = 0.2$, $\varphi = 0.5$, and $\kappa = 0$.}
        \label{plot_beta_zeta_n_500}
    \end{center}
\end{figure}

\end{document}